\documentclass[twocolumn,letterpaper,unpublished,accepted=2022-09-30]{quantumarticle}
\pdfoutput=1
\usepackage[T1]{fontenc}
\usepackage[latin9]{inputenc}
\usepackage{amsmath,amsthm,bbm}
\usepackage{amssymb}
\usepackage{graphicx,cite,appendix}
\usepackage[colorlinks=true,citecolor=blue,linkcolor=blue,urlcolor=blue]{hyperref}
\def\oper{{\mathchoice{\rm 1\mskip-4mu l}{\rm 1\mskip-4mu l}
{\rm 1\mskip-4.5mu l}{\rm 1\mskip-5mu l}}}
\def\<{\langle}
\def\>{\rangle}

\newtheorem{Proposition}{Proposition}

\usepackage{color}
\usepackage{soul,braket}
\usepackage{cancel}
\newcommand{\BH}{{\mathcal{B}(\mathcal{H}_S)}}
\usepackage{xcolor}
\usepackage[normalem]{ulem}
\newcommand{\stkout}[1]{\ifmmode\text{\sout{\ensuremath{#1}}}\else\sout{#1}\fi}

\makeatletter

\usepackage{tikz}
\usepackage{pgfplots}
\usetikzlibrary{positioning}
\usetikzlibrary{decorations.markings}
\usetikzlibrary{arrows,decorations.text}
\usetikzlibrary{backgrounds,fit,decorations.pathreplacing}

\makeatother

\usepackage{babel}

\begin{document}

\title{How to design quantum-jump trajectories via distinct master equation representations}

\author{Dariusz Chru\'sci\'nski}
\email{darch@fizyka.umk.pl}
\affiliation{Institute of Physics, Faculty of Physics, Astronomy and Informatics,
Nicolaus Copernicus University, Grudziadzka 5/7, 87-100 Toru\'{n},
Poland}

\author{Kimmo Luoma}
\email{kimmo.luoma@tu-dresden.de}
\affiliation{Institut f{\"u}r Theoretische Physik, Technische Universit{\"a}t Dresden,
 D-01062, Dresden, Germany}
\affiliation{Turku Center for Quantum Physics, Department of Physics and
Astronomy, University of Turku, FI-20014, Turun Yliopisto, Finland}

\author{Jyrki Piilo}
\email{jyrki.piilo@utu.fi}
\affiliation{Turku Center for Quantum Physics, Department of Physics and
Astronomy, University of Turku, FI-20014, Turun Yliopisto, Finland}

\author{Andrea Smirne}
\email{andrea.smirne@unimi.it}
\affiliation{Dipartimento di Fisica ``Aldo Pontremoli'', Universit{\`a} degli Studi di Milano, Via Celoria 16, I-20133 Milan, Italy}
\affiliation{Istituto Nazionale di Fisica Nucleare, Sezione di Milano, Via Celoria 16, I-20133 Milan, Italy}

\begin{abstract}

Every open-system dynamics can be associated to infinitely many stochastic pictures, called
unravelings, which have proved to be extremely useful in several contexts, both from the conceptual
and the practical point of view. Here, focusing on quantum-jump unravelings, we demonstrate that there exists inherent freedom in how to assign 
the terms of 
the underlying master equation to the deterministic and jump parts of the stochastic description, which leads to a number of qualitatively different unravelings.
As relevant examples, we show that a fixed basis of post-jump states can be selected under some definite conditions, or that the deterministic evolution can be set by a chosen
time-independent non-Hermitian Hamiltonian, even in the presence of external driving. Our approach relies on the definition of rate operators, whose positivity equips each unraveling with a 
continuous-measurement scheme and is related to a 
long known but so far not widely used
property to classify quantum dynamics, known as dissipativity. 
Starting from formal mathematical concepts, our results allow us to
get fundamental insights into open quantum system dynamics and to enrich their numerical simulations.
\end{abstract}

\maketitle

\section{Introduction}
Quantum jumps provide a powerful and insightful tool to describe the dynamics
of open quantum systems \cite{Breuer,Carmichael}.
In the quantum-jump unraveling, the open-system state satisfying an assigned
master equation is expressed as the average of -- in principle, infinitely many -- trajectories of pure states, which
consist of deterministic evolutions interrupted by discontinuous, randomly distributed jumps \cite{Dalibard1992}.
Quantum jumps have been observed in several experimental platforms \cite{Basche1995,Peil1999,Jelezko2002,Gleyzes2007,Vijay2011,Minev2019}
and are computationally convenient to solve high-dimensional
master equations,
thus being routinely used
to describe, e.g., quantum optical \cite{Plenio1998} or open many-body systems \cite{Daley2014}.
Interestingly, if the jumps are frequent and
small enough, one can recover a different stochastic description
of the open-system evolution, characterized by
diffusive trajectories \cite{Percival2002,Barchielli2009,Wiseman1993,Strunz1999, Yu1999,Luoma2020},
which have also been investigated extensively in experiments \cite{Siddiqi2013,Huard2016,Gourgy2016,Ficheux2018}.

In the standard quantum-jump method, named Monte Carlo wave function (MCWF), the state transformations induced by the jumps and their occurrence probabilities are directly fixed
by the operators and coefficients of the master equation \cite{Dalibard1992,Plenio1998}.
This allows one to associate each trajectory
with a continuous selective measurement performed on the open system \cite{Barchielli1991},
so that the master equation
can be seen as the result of the action of a non-selective observer, replacing the environment.
However, the very definition of MCWF calls for a master equation with positive coefficients.
Under some regularity conditions, this requirement is equivalent to the completely-positive(CP)-divisibility of the dynamics \cite{Laine2010,Chruscinski2011,Rivas2012}, meaning that the dynamics can be decomposed
into intermediate completely positive maps. CP-divisibility has been introduced within the context of the definition of quantum Markovianity \cite{Rivas2010,Rivas2014} and its validity implies the absence of memory effects \cite{Breuer2009,Laine2010,Breuer2016}.

While generalizations of MCWF for master equations with negative coefficients have been introduced \cite{Piilo2008,Piilo2009}, the possibility to extend
the continuous-measurement picture beyond 
the realm of
the CP-divisible evolutions has been extensively
debated~\cite{Gambetta2003,Diosi2008,Wiseman2008}.
Only very recently, a systematic approach has been put forward \cite{Jyrki-2020} 
to read the quantum-jump unraveling
of non-CP-divisible evolutions
in terms of continuous measurements. The approach relies on the definition of a proper rate operator \cite{Diosi1985,Diosi1986,Diosi1988,Gisin1990} and it has hence been named rate operator
quantum jumps (ROQJ). It is associated with a continuous-measurement scheme
applying to any positive(P)-divisible dynamics \cite{Vacchini2011,Sabrina-PRL,Wissmann2015,Breuer2016}, i.e., dynamics that can be decomposed into intermediate
maps which are positive, but not necessarily completely positive.
Such a scheme calls for an adjustment of the measurement apparatus conditioned on the previous sequence of outcomes, along the same lines of what happens in
measurement-based feedback protocols \cite{Wiseman2010,Zhang2017,Gourgy2018,Martin2020,Magrini2021}.

In this paper, we define a class of 
 jump unravelings
interpolating between and combining the advantages of MCWF and
ROQJ. We exploit the freedom in dividing any  
master equation into a deterministic and a jump part, possibly even mixing
the contributions from the Hamiltonian and the dissipative
terms. Besides its conceptual interest, this allows us to simplify the
numerical and experimental implementation of the trajectories.
On the one hand,
we can  select a fixed set of post-jump
states for a well-defined class of master equations. On
the other, we can instead set a desired time-independent, non-Hermitian
linear operator ruling the deterministic evolution, also in the
presence of an external time-dependent driving on the open system. 
Importantly, our analysis also shows that, under a 
constraint on P-divisible dynamics known as dissipativity \cite{L},
there exist (infinitely) many unravelings that can be associated with positive rate operators,
bringing along different continuous-measurement schemes.

The rest of the paper is organized as follows.
In Sec.~\ref{sec:qju}, we present the notation and the general notions of quantum-jump
unravellings that will be used throughout the paper, along with a detailed
description of the continuous-measurement scheme associated with ROQJ unravelings. 
In Sec.~\ref{sec:aco}, we introduce a whole class of rate-operator unravelings, which interpolates between
MCWF and ROQJ and allows for a rather versatile control of the resulting trajectories.
The positivity of such unravelings is discussed in Sec.~\ref{sec:pos}, where we prove that it is guaranteed
by dissipativity of the dynamics. Sec.~\ref{sec:ex} is devoted to an extended analysis of a case study,
a two-level system dynamics where CP-divisibility is broken at any time,
for which we compare the descriptions obtained via different ROQJ unravelings. 
Finally, the conclusions of our work are discussed in Sec.~\ref{sec:con}.

\section{Quantum-jump unravelings}\label{sec:qju}

We start off by recalling briefly the formalism of quantum jumps to describe the dynamics
of open quantum systems; in particular, we focus on the standard quantum-jump unraveling, i.e.,
MCWF \cite{Dalibard1992,Plenio1998}, and the recently introduced ROQJ \cite{Jyrki-2020}.

\subsection{Open quantum system dynamics}\label{sec:oqsd}

We consider the evolution of finite-dimensional open quantum systems,
given by the time-local master equation
$d \rho(t)/(d t) = \mathcal{L}_t(\rho(t))$, where $\mathcal{L}_t$
is the generator
\begin{equation}\label{GKLS}
  \mathcal{L}_t(\rho) = -i[H(t),\rho] +   \mathcal{J}_t(\rho) - \frac{1}{2}\{ \Gamma(t),\rho\} ,
\end{equation}
with
$\mathcal{J}_t(\rho) = \sum_{\alpha=1}^{N^2-1} c_\alpha(t) L_\alpha(t) \rho L_\alpha^\dagger(t)$
and $\Gamma(t)  = \sum_{\alpha=1}^{N^2-1} c_\alpha(t) L_\alpha^\dagger(t) L_\alpha(t)$,
where $N$ is the dimension of the open-system Hilbert space $\mathcal{H}_S$,
$L_\alpha(t)$ and the Hamiltonian $H(t)=H^\dag(t)$ are linear operators on $\mathcal{H}_S$,
and $c_{\alpha}(t)$ are real functions of time.
Importantly, $\mathcal{J}_t$ is a Hermiticity-preserving map
and one has
the duality relation 
\begin{equation}\label{eq:dual}
\Gamma(t) = \mathcal{J}_t^\dagger(\mathbbm{1}), 
\end{equation}
where
$\mathbbm{1}$ is the identity operator and $\mathcal{J}_t^\dagger$ is the dual
map of $\mathcal{J}_t$.
The structure of the generator follows from the trace- and Hermiticity-preservation properties
of the dynamics \cite{GKS} $\Lambda_t = T\exp\left(\int_0^t d \tau \mathcal{L}_\tau\right)$ ($T$ is the time-ordering operator);
moreover, the functions $c_\alpha(t)$ can take
negative values, yet with the resulting evolution being positive
and even completely positive \cite{Breuer,Chruscinski2010}.
On the other hand, under some regularity conditions \cite{Laine2010,Chruscinski2011,Rivas2012}, the positivity of the
coefficients, $c_\alpha(t)\geq 0$ for every $\alpha$ and $t\geq0$, is equivalent to the CP-divisibility of
the dynamics, i.e,
for any $t\geq s \geq 0$ one has the decomposition $\Lambda_t = V_{t,s}\Lambda_s$,
where the so-called propagator $V_{t,s}$ is completely positive and trace preserving;
CP-divisibility has been identified with quantum Markovianity in \cite{Rivas2010,Rivas2014}.

A crucial remark for our following analysis is that
the representation of the generator $\mathcal{L}_t$
via Eq.(\ref{GKLS}) -- with a Hamiltonian term, a Hermiticity-preserving map and 
an operator fixed by the duality relation (\ref{eq:dual}) -- is highly non-unique. In fact,
one can define a new map
\begin{equation}\label{eq:jprime}
  \mathcal{J}'_t(\rho) =  \mathcal{J}_t(\rho) + \frac{1}{2}\left(\mathbf{C}(t) \rho + \rho \mathbf{C}^\dagger(t)\right),
\end{equation}
with an arbitrary linear operator $\mathbf{C}(t)$ on $\mathcal{H}_S$, so that decomposing 
\begin{equation}\label{bho}
\mathbf{C}(t) = A(t) + i B(t)
\end{equation} 
(with $A(t)$ and $B(t)$ Hermitian operators), together with
\begin{equation}\label{eq:hprime}
  H'(t) = H(t) +\frac{1}{2} B(t), \quad \Gamma'(t) = \Gamma(t) + A(t),
\end{equation}
the same generator $\mathcal{L}_t$ can be written as
\begin{equation}\label{eq:llp}
\mathcal{L}_t(\rho) = -i[H'(t),\rho] + \mathcal{J}'_t(\rho) - \frac 12 \{\Gamma'(t),\rho\}
\end{equation}
and one still has 
\begin{equation} \label{eq:dualprime}
\Gamma'(t) = \mathcal{J}_t'^{\dagger}(\mathbbm{1});
\end{equation}
indeed, $\Gamma'(t)$ is Hermiticity-preserving.

In the following we are going to introduce a class of jump unravelings where the map 
$\mathcal{J}_t'$
and the operator $\Gamma'(t)$ along with the Hamiltonian $H'(t)$ lead to, respectively, the jump and deterministic part of the trajectories; 
different unravellings of the same dynamics 
are thus defined for each set $\left\{H'(t), \mathcal{J}_t', \Gamma'(t)\right\}$ 
such that the non-Hamiltonian terms satisfy the duality relation in 
Eq.(\ref{eq:dualprime}) and that leads to the same generator $\mathcal{L}_t$. 
Let us stress that this is at variance with other jump unravelings, such as MCWF and ROQJ, which are instead typically defined starting from the 
Lindblad operators and coefficients ($L_\alpha(t)$ and $c_\alpha(t)$),
and are thus not affected by the rewriting of the generator via Eq.(\ref{eq:jprime})-(\ref{eq:llp}).

\subsection{Monte Carlo wave function vs rate operator quantum jumps}\label{sec:mcqj}
Both the MCWF and the ROQJ unravelings consist in
piecewise deterministic processes on the set of pure states in $\mathcal{H}_S$, 
that is, they combine a deterministic time evolution and a jump process \cite{Breuer}. 
However, the specific form of the deterministic and jump parts
are different in the two kinds of unraveling and, as a consequence, the range of applicability of the two methods is different.

In the case of CP-divisible dynamics,
the master equation fixed by Eq.(\ref{GKLS}) can be unravelled by means of the MCWF
method.
The deterministic parts of the trajectories are fixed by the
non-Hermitian linear operator 
\begin{equation}\label{eq:ex0}
K(t) = H(t) - \frac{i}{2} \Gamma(t),
\end{equation}
according to
\begin{equation}\label{I}
 |\psi(t)\> \to |\psi(t+dt)\> = \frac{ (1-i K(t)dt) |\psi(t)\>}{ \|  (1-i K(t)dt) |\psi(t)\> \|},
\end{equation}
while the discontinuous parts, the jumps, are given by
\begin{equation}\label{II}
  |\psi(t)\> \to |\psi(t+dt)\> = \frac{ L_\alpha(t) |\psi(t)\>}{ \| L_\alpha(t) |\psi(t)\> \|},
\end{equation}
and
each jump occurs between $t$ and $t+d t$ with probability
\begin{equation}
p_{\psi(t),\alpha} = c_\alpha(t) \| L_\alpha(t) |\psi(t)\> \|^2 \, dt,
\end{equation}
where $d t$ is an infinitesimal time increment.
It is clear that the previous formulation
requires all the rates $c_\alpha(t)$ to be positive.

Extending the results of \cite{Diosi1985,Diosi1986,Diosi1988}, 
recently in \cite{Jyrki-2020} it has been shown that such a requirement can be weakened considerably via the definition of a different quantum jump unraveling, named ROQJ.
The latter relies on the definition of the rate operator
\begin{equation}\label{ro}
  \mathbf{W}_{\psi} =
  \sum_{\alpha=1}^{N^2-1} c_\alpha(t) (L_\alpha(t) - \ell_{\psi,\alpha}) P_\psi (L_\alpha(t) - \ell_{\psi,\alpha}(t))^\dagger,
\end{equation}
with $\ell_{\psi,\alpha}(t) = \<\psi|L_\alpha(t)|\psi\>$
and the projector $P_\psi = |\psi\> \< \psi |$.
As observed in \cite{Diosi1986,Diosi1988},
$\mathbf{W}_{\psi}$
is directly associated with the time-local generator
$\mathcal{L}_t$ and does not depend on its specific representation via
$H(t), \mathcal{J}_t$ and $\Gamma(t)$:
in fact, it can be equivalently written as
\begin{equation}\label{rom}
     \mathbf{W}_{\psi} = (\mathbbm{1} - P_\psi) \mathcal{L}_t(P_\psi) (\mathbbm{1} - P_\psi).
\end{equation}

The building block of ROQJ is the observation that $ \mathbf{W}_{\psi} \geq 0$ for any state vector $|\psi\>$ if and only if the corresponding dynamics
is P-divisible \cite{Caiaffa2017},
i.e., for any $t\geq s\geq 0$ $\Lambda_t$ can be decomposed as $\Lambda_t = V_{t,s}\Lambda_s$,
and the propagator $V_{t,s}$ is trace preserving and positive, but not necessarily completely positive;
P-divisibility, which is a significantly weaker requirement than CP-divisibility,
has been identified with quantum Markovianity in \cite{Wissmann2015,Breuer2016}.
Focusing on P-divisible evolutions,
one has then the spectral resolution 
\begin{equation}\label{eq:spectral}
\mathbf{W}_{\psi(t)} = \sum_{k=1}^N \lambda_{\psi(t),k} |\varphi_{\psi(t),k}\>\<\varphi_{\psi(t),k}|,
\end{equation} with $\lambda_{\psi(t),k} \geq 0$ for all $k$ and $t \geq 0$.
Introducing the non-Hermitian state-dependent operator
\begin{equation}\label{eq:roqjk}
  {K}_{\psi(t)} = K(t) + \Delta_{\psi(t)},
\end{equation}
with the non-linear  correction 
$$\Delta_{\psi(t)} = \frac{i}{2} \sum_{\alpha=1}^{N^2-1} c_\alpha(t) (2L_\alpha(t) \ell^*_{\psi,\alpha}   - |\ell^*_{\psi,\alpha}|^2),$$
one realizes the jump unraveling as follows \cite{Jyrki-2020}: the deterministic evolution
\begin{equation}\label{Ia}
  |\psi(t)\> \to |\psi(t+dt)\> = \frac{ (1-i K_{\psi(t)}dt) |\psi(t)\>}{ \|  (1-i K_{\psi(t)}dt) |\psi(t)\> \|}
\end{equation}
is interrupted by the sudden jumps
\begin{equation}\label{IIa}
  |\psi(t)\> \to |\psi(t+dt)\> = \frac{ V_{\psi(t),j} |\psi(t)\> }{ \| V_{\psi(t),j}  |\psi(t)\> \|},
\end{equation}
where 
\begin{equation}\label{eq:roqjv}
V_{\psi(t),j} = \sqrt{\lambda_{\psi(t),j}} |\varphi_{\psi(t),k}\>\< \psi(t)|,
\end{equation} 
and the probability that the jump $j$ occurs between $t$ and $t+dt$ is
\begin{equation}\label{eq:roqjpr}
  {p}'_{\psi(t),j} = \| V_{\psi(t),j}  |\psi(t)\> \|^2 dt =  \lambda_{\psi(t),j}  dt ;
\end{equation}
indeed the deterministic evolution occurs with probability
\begin{equation}\label{eq:roqjpdet}
 p'^{det}_{\psi} = 1 - \sum_{j=1}^N p'_{\psi,j}.
\end{equation}

Compared to MCWF, the operators $L_\alpha(t)$ and rates $c_\alpha(t)$
are replaced in ROQJ by the eigenvectors and eigenvalues of the rate operator $\mathbf{W}_{\psi(t)}$.
This is the key that allows one to formulate a well-defined unraveling,
with positive probabilities ${p}'_{\psi(t),j}$,
for the set of P-divisible dynamics, rather than just for CP-divisible dynamics as in MCWF.
For positive semigroups (for which the generator
does not depend on time, $\mathcal{L}_t = \mathcal{L}$), a jump-unravelling fully equivalent to the ROQJ has been introduced in \cite{Diosi1986,Diosi1988}, 
and further investigated (under the name of ortho-jumps) in \cite{Braun2000,Braun2002}; see also \cite{Diosi2017}. Later \cite{Caiaffa2017,Jyrki-2020}, the positivity of the rate operator has been shown to be in one-to-one correspondence with the P-divisibility of the dynamics, thus extending the definition of the unraveling to this whole class of dynamics. In this paper, we show how it is
possible to exploit the non-uniqueness of the master equation representation to define
a whole class of positive unravelings.

Both MCWF and ROQJ can be extended to deal with, respectively, non-CP-divisible and
non-P-divisible evolutions, via the use of reversed jumps connecting different
trajectories \cite{Piilo2008,Piilo2009,Jyrki-2020}.
However, ROQJ applies to general not necessarily P-divisible evolutions \cite{Jyrki-2020}, and it can be used to treat certain non-CP-divisible dynamics \cite{ENM,Chruscinski2015,Nina} where the non-Markovian version of MCWF cannot be used.

Finally, we stress that the ROQJ involves the diagonalization of the rate operator,
which adds complexity to the actual implementation of the unraveling. 
On the other hand, such a diagonalization concerns an $N \times N$ matrix and it
is thus significantly simpler, for example, than the direct
diagonalization of the $N^2 \times N^2$ matrix associated with the generator of the master equation. In addition, crucially, the diagonalization in ROQJ
needs not to be done at each time step, contrary to what would happen for example in
a diffusive unraveling \cite{Caiaffa2017}.
Rather, one has to diagonalize the rate operator only if a jump occurs and  jumps
are indeed rare events, whose probability is proportional to the infinitesimal time step $d t$.
The key point is that one can fix
whether a jump occurs or not by looking at the deterministic part
of the evolution only, as the probability of having any jump is
\begin{align}\label{eq:prequ}
\sum_j {p}'_{\psi(t),j}& = \sum_j \lambda_{\psi(t),j}  dt = \mbox{Tr}\left\{\mathbf{W}_{\psi(t)} \right\} dt,
\end{align}
where $\mbox{Tr}$ denotes the trace,
which leads us to (neglecting as usual terms of order $dt^2$)
\begin{equation}\label{eq:qu}
\sum_j {p}'_{\psi(t),j} = 1- \|  (1-i K(t)dt) |\psi(t)\> \|^2.
\end{equation}
Therefore, since we know the operator $K(t)$ giving the deterministic evolution,
we need to construct the rate operator and diagonalize it only at those (rare) times when a jump takes place.
Note that an analogous reasoning applies
to the rate operators $\mathbf{R}_{\psi}$ described in the following.

\subsection{Continuous-measurement scheme}\label{sec:cms}

Both in MCWF \cite{Barchielli1991} and in ROQJ \cite{Jyrki-2020} unraveling the
trajectories can be seen as due to a continuous measurement on the open system.
The state transformations and corresponding probabilities can be associated
with a quantum instrument \cite{Heinosaari2012}, mapping the set of outcomes 
into a set of open-system completely positive trace non-increasing maps
that sum up to a trace preserving map, but, crucially, in ROQJ continuous-measurement schemes 
can be defined for the whole set of
P-divisible dynamics.

Let us in fact consider the rate operator $\mathbf{W}_{\psi}$ of a P-divisible dynamics,
and the corresponding jump operators $\left\{V_{\psi,j}\right\}_{j=1, \ldots N}$
and jump probabilities $\{p'_{\psi,j}\}_{j=1, \ldots N}$ defined 
as in Eqs.(\ref{eq:roqjv}) and (\ref{eq:roqjpr}). In addition, let us denote as $V_{\psi,\emptyset}$
the operator obtained from the first-order evolution associated with the effective Hamiltonian in Eq.(\ref{eq:roqjk}) according to
\begin{equation}\label{eq:mattbis}
V_{\psi,\emptyset} = \left(1-i K_{\psi}dt\right) \ket{\psi}\bra{\psi},
\end{equation}
and as $p'_{\psi,\emptyset} = p'^{det}_{\psi}$ the probability of having a deterministic
evolution according to Eq.(\ref{eq:roqjpdet}).
It is then readily seen that these operators and probabilities 
correspond to the state transformations and associated probabilities
of a well-defined quantum instrument,
for any fixed $\ket{\psi}\in\mathcal{H}_S$. 

Take in fact the set of outcomes 
$\mathcal{O}=\left\{a, \emptyset, j\right\}_{j=1, \ldots N}$
and the maps $\left\{\mathcal{I}_{\psi,j} \right\}_{j=a, \emptyset, 1, \ldots N}$
whose action on a generic state $\rho$ on $\mathcal{H}_S$ is given by
\begin{equation}\label{eq:roqjins}
\mathcal{I}_{\psi,j}(\rho) = V_{\psi,j} \rho V^{\dag}_{\psi,j} \quad j=a, \emptyset, 1, \ldots N,
\end{equation}
where we introduced also the (auxiliary) outcome $a$, along with the operator 
$V_{\psi, a} = \mathbbm{1} -  \ket{\psi}\bra{\psi}$.
Now, for any fixed $\ket{\psi}\in\mathcal{H}_S$, the maps in Eq.(\ref{eq:roqjins}) are indeed completely positive
and trace non-increasing and they sum up to a trace preserving map, since
$\sum_{j=a, \emptyset, 1, \ldots N} V^\dag_{\psi,j} V_{\psi,j}=\mathbbm{1}$,
as follows from Eqs.(\ref{eq:spectral}), (\ref{eq:roqjv}), (\ref{eq:prequ}) and (\ref{eq:qu}).
Thus, the state transformations
and associated probabilities
\begin{equation}\label{eq:siono}
\rho \mapsto \frac{\mathcal{I}_{\psi,j}(\rho)}{\mbox{Tr}\left\{ \mathcal{I}_{\psi,j}(\rho)\right\}}
\qquad \mathcal{P}_{\psi,j} = \mbox{Tr}\left\{ \mathcal{I}_{\psi,j}(\rho)\right\}
\end{equation}
correspond to the post-measurement states and probabilities 
of the measurement with outcomes $\mathcal{O}$
and described by the quantum instrument 
$\left\{\mathcal{I}_{\psi,j} \right\}_{j=a, \emptyset, 1, \ldots N}$.
But if we now focus on the action of the maps on the pure state 
$\rho = \ket{\psi}\bra{\psi}$,
we see that the state transformations and probabilities
in Eq.(\ref{eq:siono}) for $j = 1, \ldots, N$
are nothing else than, respectively, the state after the jump as in Eqs.(\ref{eq:roqjv}) 
and the jump probability in (\ref{eq:roqjpr}), $\mathcal{P}_{\psi,j}=p'_{\psi,j}$, 
while for $j = \emptyset$ we get the deterministic evolution in Eq.(\ref{Ia})
and the probability in Eq.(\ref{eq:roqjpdet}), $\mathcal{P}_{\psi,\emptyset}=p'_{\psi,\emptyset}$, 
where the latter follows from the identity in Eq.(\ref{eq:qu});
on the other hand, the (auxiliary) outcome $a$ occurs with zero probability, $\mathcal{P}_{\psi,a}=0$. 
The instrument in Eq.(\ref{eq:roqjins}) can be realized, for example, by $N$ counters surrounding the system
and parametrized by the index $j$: a click of the $j$-th counter indicates that the system
jumps to the $j$-th eigenstate of $\mathbf{W}_{\psi}$.
Importantly, also the scenario where no counter clicks corresponds to a measurement
performed on the system, with null result $\emptyset$ \cite{Wiseman1996}
and, as said, resulting in the evolution fixed by Eq.(\ref{eq:mattbis}). 

If we now consider any trajectory of the jump unraveling\footnote{On more precise mathematical terms, one should consider a stochastic pure state, that is a stochastic process with values in $\mathcal{H}_S$,
defined by a counting process
whose trajectories correspond to the sequences of jumps \cite{Barchielli1991,Jyrki-2020}.} defined by 
$\mathbf{W}_{\psi}$ 
in Sec.\ref{sec:mcqj}, we can repeat the reasoning above for any infinitesimal time $dt$
starting from any pure initial state $\ket{\psi_0}$,
which means that such a trajectory is equivalently obtained as the result of a continuous measurement
performed on the quantum system associated with $\mathcal{H}_S$
and described by the quantum instrument in Eq.(\ref{eq:roqjins}).
As a consequence, the open-system dynamics fixed by the master equation (\ref{GKLS})
and resulting from the average over the trajectories of the unraveling is equivalently obtained
as the consequence of a continuous non-selective monitoring of the system at hand, meaning that
one is continuously measuring the system via the measurement apparatus described
by the instrument in Eq.(\ref{eq:roqjins}), but without selecting the system according to the 
measurement outcomes. 

Crucially, the instrument to be applied on each trajectory at any time $t$ depends on 
the (stochastic) state $\ket{\psi(t)}$ or, equivalently, on the previous sequence of events,
which is what allows us to introduce a proper continuous-measurement scheme
in the presence of P-divisible, but not necessarily CP-divisible dynamics \cite{Jyrki-2020}.
On a practical level, this means that the actual realization of the continuous-measurement scheme
described in this section calls for a continuous adjustment of the measurement
apparatus monitoring the system, depending on the sequence of outcomes.
While the strategy to implement this procedure 
will crucially depend on the specific experimental platform at hand,
we can already point to a correspondence with 
the adaptive
approaches that characterize, for example, 
the measurement-based feedback
strategies described in~\cite{Gourgy2018,Martin2020,Magrini2021}, where the measurement basis is changed dynamically
and according to the previous sequence of measurement outcomes.

\section{A new class of rate operator quantum jumps}\label{sec:aco}
We are now ready to present a novel class of quantum-jump unravelings, which 
is based on a family of rate operators defined starting from Eqs.(\ref{eq:jprime})-(\ref{eq:hprime}). This class combines features of both MCWF,
allowing for a linear effective Hamiltonian, and the ROQJ discussed above, being positive for dynamics
with possibly negative rates. In addition, we show how the freedom in choosing the rate operator
allows us to control and manipulate some basic properties of the trajectories of the unraveling,
which can be useful for the numerical simulation of the dynamics and the actual
experimental implementation of the jumps.

\subsection{Definition of the unravelings}\label{sec:dotu}
The basic observation to define the class of ROQJs is that
$\mathbf{W}_\psi$ in Eq.(\ref{rom})
can also be written as
\begin{equation}\label{WR}
  \mathbf{W}_\psi = (\mathbbm{1} - P_\psi) \mathbf{R}_\psi (\mathbbm{1} - P_\psi),
\end{equation}
where
\begin{equation}\label{R}
  \mathbf{R}_{\psi(t)} = \mathcal{J}_t(|\psi(t)\>\<\psi(t)|).
\end{equation}
Now, if $ \mathbf{R}_{\psi(t)} \geq 0$
for any $|\psi(t)\>$ one can define a jump unraveling that merges (\ref{eq:ex0})-(\ref{I}) 
with (\ref{IIa})-(\ref{eq:roqjv});
that is, the deterministic evolution is governed by (\ref{I}), but the jumps are realized via
\begin{equation}\label{IIb}
  |\psi(t)\> \to |\psi(t+dt)\> = \frac{ R_{\psi(t),k}  |\psi(t)\>}{ \| R_{\psi(t),k} |\psi(t)\> \|} ,
\end{equation}
with 
\begin{equation}\label{eq:pre1}
R_{\psi(t),k} = \sqrt{ r_{\psi(t),k}}  |\phi_{\psi(t),k}\>\<\psi(t)|,
\end{equation}
and the probability that the jump $k$ occurs between $t$ and $t+dt$ reads now
\begin{equation}\label{IIc}
  {p}''_{\psi(t),k} = \| R_{\psi(t),k}  |\psi(t)\> \|^2 dt =  r_{\psi(t),k} \, dt;
\end{equation}
indeed, $r_{\psi(t),k}$ and $|\phi_{\psi(t),k}\>$ are the eigenvalues and eigenvectors
of $\mathbf{R}_{\psi(t)}$, i.e.,
\begin{equation}\label{eq:ex1}
\mathbf{R}_{\psi(t)} = \sum_{k=1}^N r_{\psi(t),k} |\phi_{\psi(t),k}\>\<\phi_{\psi(t),k}|
\end{equation}
and the deterministic evolution will occur with probability
$p''^{det}_{\psi(t)} = 1 - \sum_{k=1}^N  {p}''_{\psi(t),k}$.
General conditions guaranteeing the positivity of $\mathbf{R}_{\psi(t)}$ will be presented
in Sec.~\ref{sec:pos}.
We stress that whenever $\mathbf{R}_{\psi(t)}$ is positive 
the analysis performed in Sec.\ref{sec:cms}
can be readily adapted to the unraveling fixed by this rate operator, so that a continuous-measurement
scheme can be defined via a quantum instrument
with maps $\mathcal{I}_{\psi,j}$ as in Eq.(\ref{eq:roqjins}),
where $\{V_{\psi,j}\}_{j=1,\ldots N}$ are indeed replaced by $\{R_{\psi,k}\}_{k=1,\ldots N}$,
while the operator $V_{\emptyset,j}$ corresponding to the outcome $\emptyset$ and associated 
with the deterministic evolution is now fixed by the effective Hamiltonian in Eq.(\ref{eq:ex0}),
compare with Eq.(\ref{eq:mattbis}).

Before proving that the construction fixed by Eqs.(\ref{R})-(\ref{eq:ex1}) 
provides us with a well-defined unraveling
of the master equation~(\ref{GKLS}) whenever the rate operator is positive, 
let us stress that it avoids the non-linear correction in the non-Hermitian
operator fixing the deterministic evolution, still potentially allowing
for positive probabilities in the case of non-CP-divisible evolutions;
${p}''_{\psi(t),k}\geq 0$ is in fact equivalent to the requirement
that the map $\mathcal{J}_t$ is positive for all $t \geq 0$,
which is significantly weaker than CP-divisibility as we will see in Sec.~\ref{sec:pos}.
Moreover, the definition of the rate operator
$\mathbf{R}_{\psi(t)}$ in Eq.(\ref{R}) does depend on the specific choice
of the map $\mathcal{J}_t$ in the representation of the generator $\mathcal{L}_t$
as in Eq.(\ref{GKLS}); as a consequence, Eqs.(\ref{I}), (\ref{IIb}) and (\ref{IIc})
define a whole family of unravelings,
corresponding to different choices of the operator $\mathbf{C}(t)$ in 
Eqs.(\ref{eq:jprime})-(\ref{eq:hprime}). 
We denote these unravelings with
$\mathbf{R}$-ROQJ, specifying the explicit form of the rate operator $\mathbf{R}$
when needed; furthermore, we will refer to the ROQJ unraveling discussed 
in Sec.~\ref{sec:mcqj}
as $\mathbf{W}$-ROQJ.

\begin{proof}
Given the pure state $|\psi(t)\>\<\psi(t)|$ at time $t$, the deterministic
evolution will occur with probability 
\begin{align}
p''^{det}_{\psi(t)} &= 1 - \sum_{k=1}^N  {p}''_{\psi(t),k}
= 1 - \mbox{Tr}\left\{\mathbf{R}_{\psi(t)} \right\} dt,
\end{align}
where in the last equality we used Eq.(\ref{eq:ex1}).
But using Eq.(\ref{R}) along with 
\begin{equation}\label{eq:exmu}
\Gamma(t) = \mathcal{J}_t^\dagger(\mathbbm{1}),
\end{equation}
we get
\begin{align}
p''^{det}_{\psi(t)} &= 1 - \mbox{Tr}\left\{ \mathcal{J}_t(|\psi(t)\>\<\psi(t)|) \right\} dt \notag\\
&= 1- \<\psi(t)|\Gamma(t)|\psi(t)\> dt,\label{eq:ex2}
\end{align}
which highlights the role of the duality relation in Eq.(\ref{eq:exmu})
to express the probability of the deterministic evolution in terms of both the 
jump part of the master equation $\mathcal{J}_t$ and the term $\Gamma(t)$
entering into the non-Hermitian operator in Eq.(\ref{eq:ex0}). 
The deterministic evolution will map the pure state $|\psi(t)\>\< \psi(t)|$
into the pure state (see Eq.(\ref{I}))
\begin{align}
 &\frac{ (1-i K(t)dt) |\psi(t)\>\< \psi(t)| (1+i K^\dagger(t)dt)}{ \|  (1-i K(t)dt) |\psi(t)\> \|^2} \notag\\
 &= \frac{ (1-i K(t)dt) |\psi(t)\>\< \psi(t)| (1+i K^\dagger(t)dt)}{p''^{det}_{\psi(t)}},\label{eq:ex3}
\end{align}
where we used Eqs.(\ref{eq:ex0}) and (\ref{eq:ex2}) in the denominator (neglecting the terms of order $d t^2$).

On the other hand, as said, given the state $|\psi(t)\>\<\psi(t)|$ at time $t$, 
we will have the jump described by $R_{\psi(t),k}$ in Eq.(\ref{eq:pre1}), i.e. (compare
with Eq.(\ref{IIb}))
\begin{equation}\label{eq:ex4}
|\psi(t)\>\<\psi(t)| \to  \frac{ R_{\psi,k}  |\psi(t)\>\<\psi(t)| R^\dagger_{\psi,k}}{ \| R_{\psi,k} |\psi(t)\> \|^2}
= |\phi_{\psi(t),k}\>\<\phi_{\psi(t),k}|,
\end{equation}
with probability as in Eq.(\ref{IIc}).

All in all, if we average the state at time $t+dt$ over the trajectories where the state
at time $t$ is $|\psi(t)\>\<\psi(t)|$, we get the mixture of the states obtained via the deterministic evolution or one of the jumps,
each weighted with the corresponding probability: 
using Eqs.(\ref{eq:ex3}), (\ref{eq:ex4}) and (\ref{IIc}), such a (conditioned) average corresponds to
\begin{eqnarray}
&&p''^{det}_{\psi(t)}  \frac{ (1-i K(t)dt) |\psi(t)\>\< \psi(t)| (1+i K^\dagger(t)dt)}{p''^{det}_{\psi(t)}}
\nonumber\\
&&+ \sum_{k=1}^N  r_{\psi(t),k} |\phi_{\psi(t),k}\>\<\phi_{\psi(t),k}|\, dt \nonumber\\
&&= |\psi(t)\>\< \psi(t)| - i \left[ H,|\psi(t)\>\< \psi(t)| \right] dt \nonumber\\
&&- \frac 12 \{\Gamma(t), |\psi(t)\>\< \psi(t)| \} dt +
 \mathcal{J}_t(|\psi(t)\>\<\psi(t)|) dt,
\end{eqnarray}
where the equality is due to Eqs.(\ref{eq:ex0}), (\ref{R}) and (\ref{eq:ex1}), and we neglected the terms of order $dt^2$.
Finally, we perform a second average, this time with respect to the pure states
$|\psi(t)\>\<\psi(t)|$ we fixed at time $t$, so that we get the average over all the trajectories; the previous expression then yields the first order expansion of the time-local master equation 
$d \rho(t)/(d t) = \mathcal{L}_t(\rho(t))$ with $\mathcal{L}_t$ as in Eq.(\ref{GKLS}),
which concludes the proof.
\end{proof}

Indeed, the proof above does not depend on the specific representation of the generator
$\mathcal{L}_t$  and then it holds for any 
choice of the operator $\mathbf{C}$, 
as long as the corresponding rate operator $\mathbf{R}_{\psi(t)}$ is positive. 
Going beyond the results in \cite{Jyrki-2020}, we have thus defined
a new class of unravelings of the master equation (\ref{GKLS}) for P-divisible dynamics, which 
are equipped with an associated continuous-measurement
scheme, see the discussion after Eq.(\ref{eq:ex1}). Besides including the simpler linear effective Hamiltonian that characterizes the 
MCWF method and connecting the positivity of the rate operator with a significant property of the dynamics, as proved in Sec.\ref{sec:pos}, the use of distinct master-equation representations
allows us to control the deterministic and jump parts
of the unraveling in a versatile way, as we are now going to show.

\subsection{$\mathbf{R}$-ROQJ with fixed post-jump states}
\label{sec:fpj}

First of all,  a proper choice of $\mathbf{C}$ can be used to simplify to a significant extent the jumps in the unraveling. 
An explicit condition for this can be derived for two-dimensional open quantum systems,
i.e., $\mathcal{H}_S = \mathbbm{C}^2$.

Given any time-local generator $\mathcal{L}_t$ on the set of linear
operators $\mathcal{B}(\mathbbm{C}^2)$ in the GKLS form
as in Eq.(\ref{GKLS}),
$\mathcal{J}_t$ can be represented via the 16 parameters $J^{k l}_{i j}(t)$ defined by
\begin{equation}\label{eq:166}
J^{k l}_{i j}(t) =\bra{i}\mathcal{J}_t[\ket{k}\bra{l}]\ket{j}  \quad i,j=1,2,
\end{equation}
with respect to any orthonormal basis  $\left\{\ket{1},\ket{2}\right\}$;
note that the Hermiticity-preservation condition implies
$\left(J^{k l}_{i j}(t)\right)^* = J^{l k}_{j i}(t)$.
Now, if there is a basis such that
\begin{equation}
J^{k l}_{i j}(t)=(J^{k l}_{i j}(t))^*
\quad \forall i,j,k,l=1,2,\label{eq:cond2}
\end{equation}
i.e., the Choi matrix associated with $\mathcal{J}_t$ is real,
and the deterministic evolution in-between the jumps does not introduce a relative phase
when acting on the basis elements, it is possible to define a rate operator $\mathbf{R}_{\psi(t)}$ via Eqs.(\ref{eq:jprime}) and (\ref{R}) 
whose eigenvectors are independent from the pre-jump state $\ket{\psi(t)}$.
This implies that the post-jump states
are always the same, so that the trajectories are fixed by at most 3 deterministically-evolving
states, which results in a strong simplification of both the numerical simulation
and the experimental implementation of the corresponding unraveling,
as will also be shown by means of an example in Sec.\ref{sec:ex}.
More explicitely, we have in fact the following.

\begin{Proposition}\label{prop:jj}
Given the master equation (\ref{GKLS}) for $N=2$, 
if there is an orthonormal basis of $\mathbbm{C}^2$, denoted as 
$\left\{\ket{\varphi_1},\ket{\varphi_2}\right\}$, and a time-dependent linear operator $\mathbf{C}(t) \in \mathcal{B}(\mathbbm{C}^2)$
with matrix representation in this basis 
\begin{align}
&\mathbf{C}(t) = \label{eq:a}\\
&\left( \begin{array}{cc} 
J^{1 1}_{2 2}(t)-J^{1 1}_{1 1}(t) + i x(t) & y(t) \\
y(t)+2(J^{1 2}_{1 1}(t)-J^{1 2}_{2 2}(t)) & J^{2 2}_{1 1}(t)-J^{2 2}_{2 2}(t) + i x(t) 
\end{array} \right) \notag
\end{align}
with $x(t), y(t)$ two real functions,
such that
\begin{itemize}
\item the jump operator $\mathcal{J}'_t$ fixed by Eq.(\ref{eq:jprime}) satisfies Eq.(\ref{eq:cond2});
\item if the state $\ket{\psi(t)}$ at time $t$ is in the form
\begin{equation}
\ket{\psi(t)} = c(t)\ket{\varphi_1} \pm \sqrt{1-c(t)^2}\ket{\varphi_2} \label{eq:cond3}
\end{equation}
for some $-1\leq c(t) \leq 1$,
the deterministic evolution fixed by Eqs.(\ref{eq:hprime}), (\ref{eq:ex0}) (with $H'$
instead of $H$ and $\Gamma'$ instead of $\Gamma$) and (\ref{I})
implies that the state at time $t+dt$ is in the form 
\begin{eqnarray}
\ket{\psi(t+dt)}= a(t, dt) \ket{\varphi_1} \pm \sqrt{1-a(t, dt)^2}\ket{\varphi_2} \nonumber\\ \label{eq:cond1}
\end{eqnarray}
for some $-1\leq a(t, dt)\leq 1$;
\item the rate operator $\mathbf{R}_{\psi(t)} = \mathcal{J}'_t(|\psi(t)\>\<\psi(t)|)$
is positive for any state $\ket{\psi(t)}$ as in Eq.(\ref{eq:cond3});
\end{itemize}
and we restrict to
\begin{itemize}
\item initial states $\ket{\psi_0}=\ket{\psi(0)}$ in the form as in Eq.(\ref{eq:cond3}),
\end{itemize}
then there exists a jump unraveling
that involves only the 3 families of states 
$\left\{\ket{\psi_0(t)}, \ket{\varphi_+(t,s)},\ket{\varphi_-(t,s)}\right\}$, which are deterministically evolved
from $\left\{\ket{\psi_0}, \ket{\varphi_\pm}=\frac{1}{\sqrt{2}}(\ket{\varphi_1}\pm\ket{\varphi_2})\right\}$ 
via
\begin{align}
D(t,s) = & T \exp\left(-i\int_s^t d \tau \left(H'(\tau) -\frac{i}{2}\Gamma'(\tau)\right)\right). \label{eq:det}
\end{align}
\end{Proposition}
The conditions in Eqs.(\ref{eq:cond3}) and (\ref{eq:cond1}) mean that the unraveling
will involve exclusively pure states without a relative phase in the basis 
$\left\{\ket{\varphi_1},\ket{\varphi_2}\right\}$, if this is the case at the initial time; this is also why the positivity requirement on the rate operator $\mathbf{R}_{\psi(t)}$ can be now restricted only to states of this form. In practice, the deterministic evolution between two times $s$ and $t$
will be given by the operator in Eq.(\ref{eq:det}), which indeed corresponds
to the effective non-Hermitian Hamiltonian fixed by Eq.(\ref{eq:hprime})
for the operator $\mathbf{C}(t)$ as in Eq.(\ref{eq:a}). The jumps at time $t$, given by the eigenvectors
of $\mathbf{R}_{\psi(t)}$, will always end up 
in one of the eigenvectors $\left\{\ket{\varphi_+},\ket{\varphi_-}\right\}$,
and then all in all the trajectories will consist of piecewise deterministic evolutions
among $\left\{\ket{\psi_0(t)}, \ket{\varphi_{+}(t,s)},\ket{\varphi_-(t,s)}\right\}$.
Note that the probabilities of having the jumps between time $t$ and $t+dt$
are fixed by the eigenvalues of $\mathbf{R}_{\psi(t)}$, so that
different matrices $\mathbf{C}(t)$ in Eq.(\ref{eq:a}) will generally define
distinct unravelings,
with the same trajectories but different associated probabilities.

In addition, the existence of a class of ROQJ can be exploited to 
control to a large extent the deterministic evolution in-between the jumps. 
From Eqs.(\ref{bho}), (\ref{eq:hprime}) and (\ref{eq:ex0}) we see how
any effective non-Hermitian Hamiltonian can be enforced by means of a proper choice
of $\mathbf{C}$, which will also
result in a modified set of jump states and probabilities, according to Eq.(\ref{eq:jprime}), (\ref{R})
and (\ref{IIc}).
As will be shown by means of a significant example in Sec.\ref{sec:ex}, this choice can be made
while keeping the positivity of the rate operator and it is particularly convenient for driven systems.
We stress that the definition of different partitions between
the deterministic and the jump parts of the unraveling, starting from distinct representations
of the same master equation (\ref{GKLS}) by virtue of Eqs.(\ref{eq:jprime})-(\ref{eq:llp}), 
is not encompassed by the possibility to define different unravelings that is routinely used 
in MCWF, which
relies on the invariance of the generator under unitary
transformations of the set of Lindblad operators \cite{Breuer}
and would not affect the partition of the generator as dictated
by Eq.(\ref{eq:llp}) and (\ref{eq:dualprime}); compare with the remark at the end of Sec.\ref{sec:oqsd}.

\section{Positivity of the rate operators}\label{sec:pos}
We can now complete the definition of the $\mathbf{R}$-ROQJ by deriving
general conditions ensuring that it is
associated with a positive rate operator. Besides justifying the construction
in Eqs.(\ref{IIb}) and (\ref{IIc}), this also guarantees the existence of a fully consistent
continuous-measurement picture as discussed in Sec.\ref{sec:cms}.

Clearly, if the evolution is CP-divisible, not only both MCWF and $\mathbf{W}$-ROQJ
are well defined, but one can always find a completely positive map $\mathcal{J}_t$ giving rise to a positive
$\mathbf{R}$-ROQJ.
Instead, any P-divisible evolution does guarantee positivity of $ \mathbf{W}_{\psi(t)}$, but not necessarily the positivity of $ \mathbf{R}_{\psi(t)}$. Still, we 
note that P-divisibility constrains the number of possible negative eigenvalues
of the rate operator $\mathbf{R}_{\psi(t)}$.
We have in fact the following:
\begin{Proposition} For any P-divisible evolution, $\mathbf{R}_{\psi(t)}$ can have at most one negative eigenvalue.
\end{Proposition}
\begin{proof}
The proof easily  follows from the 
min-max principle for Hermitian matrices: let $A$ be a Hermitian $n \times n$ matrix, and let
$$   \lambda_n \geq \lambda_{n-1} \geq \ldots \geq \lambda_1 $$
be the real eigenvalues of $A$. Then
\begin{equation}\label{mM}
  \lambda_k = \max_{\Sigma}\, \min_{x \in \Sigma} \langle x |A| x \rangle
\end{equation}
where $x$ is normalized, and $\Sigma$ is a $(n-k+1)$-dimensional subspace of $\mathbb{C}^n$.

Now, since $ \mathbf{W}_\psi = (\mathbbm{1} - P)  \mathbf{R}_\psi  (\mathbbm{1}- P)  \geq 0$, one has for any $x \in \Sigma = (\mathbbm{1} - P)\mathbb{C}^n$
$$   \langle x |  \mathbf{R}_\psi |x \rangle =
  \langle x |  \mathbf{W}_\psi |x \rangle \geq 0 .$$
Hence,  from (\ref{mM}) one finds $\lambda_2 \geq 0$ (since in (\ref{mM}) one maximizes over all $\Sigma$s), and hence only $\lambda_1$ may be negative.
\end{proof}

\noindent We stress that the proof does not depend
on the specific $\mathcal{J}_t$ used to represent the generator $\mathcal{L}_t$
and hence applies to any $\mathbf{R}$-ROQJ.

Now, since $\mathbf{R}_{\psi(t)}$ does depend upon the specific 
$\mathcal{J}_t$, see Eq.(\ref{R}), a natural strategy to ensure the positivity of $\mathbf{R}_{\psi(t)}$ arises: is it possible to use the freedom introduced by (\ref{eq:jprime}) so that the rate operator defined in terms of $\mathcal{J}'_t$ is positive?
Interestingly, in all the examples we analyzed this is the case; 
however, we can give a positive answer in generality only for a class of P-divisible evolutions that enjoys an additional property. 
In the Heisenberg picture P-divisibility means that the propagators $V_{t,s}^\dagger$ are positive and unital. It is well known \cite{Paulsen,Stormer} that 
this implies the Kadison-Schwarz inequality
\begin{equation}\label{KS}
  V_{t,s}^\dagger(X^\dagger X) \geq V_{t,s}^\dagger(X^\dagger) V_{t,s}^\dagger(X) ,
\end{equation}
for all Hermitian operators $X=X^\dagger \in \BH$. A more restricted class of evolutions is thus defined by propagators that satisfy (\ref{KS}) for all $X$, not necessarily Hermitian. 
This leads us to the identification of the condition ensuring the positivity of 
the $\mathbf{R}$-ROQJ:
\begin{Proposition}\label{prop:ks} 
If the propagators of the dynamics satisfy (\ref{KS}) for all $X \in \BH$, there exists a representation 
of the master equation~(\ref{GKLS}) with a positive map $\mathcal{J}_t$, which we denote as $\mathbf{J}_t$.
\end{Proposition}
The proof, which is reported in Appendix \ref{app:propks}, is based on some techniques introduced by Lindblad in his seminal paper \cite{L},
based on the fact that the Kadison-Schwarz inequality may be rephrased by the following condition for the time-local generator (in the Heisenberg picture):
\begin{equation}\label{KS-L}
    \mathcal{L}^\dagger_t(X^\dagger X) \geq \mathcal{L}^\dagger_t(X^\dagger)X + X^\dagger \mathcal{L}^\dagger_t(X) ,
\end{equation}
again for all $X \in \BH$ and $t \geq 0$. The generators
satisfying Eq.(\ref{KS-L}) are called dissipative \cite{L}. All in all, the dissipativity condition is equivalent to P-divisibility if one restricts to Hermitian operators $X$; however,  assuming that
the condition (\ref{KS-L}) holds for all $X \in \BH$  the corresponding evolution is not only P-divisible but in addition the  propagator $V_{t,s}$ satisfies~(\ref{KS}). 

From the physical point of view, the dissipativity condition can be understood with the following remark: let $\varrho(t)$ be an instantaneous invariant state, i.e. $\mathcal{L}_t(\varrho(t))=0$. One has ${\rm Tr}(\varrho(t) \mathcal{L}^\dagger_t(X^\dagger X)) =  {\rm Tr}(\mathcal{L}_t(\varrho(t)) (X^\dagger X)) = 0$, and hence introducing the following inner product in the space of Hermitian operators 
$$
(X,Y)_\varrho = {\rm Tr}( \varrho X^\dagger Y)
$$ 
the formula (\ref{KS-L})  implies \cite{Rivas2012}
\begin{equation}\label{Diss-product}
 {\rm Re}\, (X, \mathcal{L}^\dagger_t(X))_{\varrho(t)} \leq 0 ,
\end{equation}
for any operator $X \in \BH$. In particular, if $|i(t)\rangle$ defines an eigenbasis of $\varrho(t)$, then taking $X =|i(t)\rangle\langle j(t)|$ the condition \eqref{Diss-product} clearly shows that P-divisibility provides constraints for populations $(i=j)$, whereas dissipativity is 
more restrictive, providing constraints also for coherences $(i \neq j$).

Summarizing, the condition in Eq.(\ref{KS}), or equivalently Eq.(\ref{KS-L}), for all $X \in \BH$ allows for a representation of (\ref{GKLS}) with a positive map $\mathbf{J}_t$.
In other terms, dissipativity guarantees the existence of a rate operator
$\mathbf{R}_{\psi(t)}$, defined from $\mathbf{J}_t$ by Eq.(\ref{R}), such that
the unraveling given by Eqs.(\ref{I}), (\ref{IIb}) and (\ref{IIc}) is well-defined,
as $\mathbf{R}_{\psi(t)} \geq 0$ for any $\psi(t)$.
As follows from the discussion in Sec.\ref{sec:cms} and \ref{sec:dotu}, this also means
that dissipativity ensures the possibility to obtain the dynamics as the result of a continuous
non-selective measurement on the system at hand.
It should be stressed that the representation with a positive map $\mathbf{J}_t$ is not unique,
corresponding to a non-unique continuous-measurement scheme, as we will
see explicitly in the examples in the next section.
The use of the $\mathbf{R}$-ROQJ allows thus for a versatile definition
of the deterministic and jump parts of the unraveling that, along with the dissipativity
of the dynamics, result in distinct experimental procedures leading to 
the detection of trajectories associated with different unraveling of the same master equation.

\section{Eternally non-Markovian qubit master equation}\label{sec:ex}
To explore in an explicit case study the different possible $\mathbf{R}$-ROQJ unravelings
and compare them among each other, as well as with
the $\mathbf{W}$-ROQJ, we consider the two-level system dynamics
fixed by the following generator
\begin{equation}\label{ENM} 
\mathcal{L}_t(\rho) = i\frac{b(t)}{2}[\sigma_z,\rho]+
\frac 12
\sum_{k=1}^3 \gamma_k(t) (\sigma_k \rho \sigma_k - \rho),
\end{equation} 
where $\sigma_k$ are the Pauli spin operators. 
Such a master equation has been studied extensively in the literature \cite{ENM,Chruscinski2015,Nina},
since, despite its simplicity, it possesses several interesting features. 
The open-system evolution fixed by Eq.(\ref{ENM})
can arise as due to an average over randomly distributed unitary evolutions
\cite{Chruscinski2015}
or as the classical mixture of Markovian
dephasing dynamics in three different directions \cite{Nina}.
In particular,
for some choices of the rates $\gamma_k(t)$ the resulting dynamics is P-divisible, while
CP-divisibility is broken for any $t>0$. This kind of dynamics is usually referred to 
as eternally non-Markovian,
indicating that the backflow of information to the open system witnessed
by the negativity of the decay rate continues for the whole evolution;
hence, any Markovian limit, even in the asymptotic time scale, is precluded.
Most importantly for our purpose,
eternally non-Markovian dynamics cannot be treated by means of the MCWF at any time of the evolution.
Even more, also a powerful non-Markovian generalization of MCWF, named
non-Markovian quantum jumps \cite{Piilo2008,Piilo2009}, cannot be applied
in this case since CP-divisibility is violated from the very beginning of the dynamics.
The general ingredients for the implementation of the jump methods, including the relevant numerical aspects, can be found from references \cite{Plenio1998,Daley2014,Molmer1996}.

It is important to stress that the smaller the effective ensemble size 
$N_{eff}$ the more efficiently the simulations can be implemented and optimized. 
This is because it is enough to generate $N_{eff}$ state vector evolutions and decide $N$ times at each time step whether the jump happened or not; 
here, $N$ is the size of the total ensemble. This allows us to track how many members of the total ensemble are in each of the $N_{eff}$ different members of the effective ensemble, 
avoiding repetitions of identical evolutions in a given interval of time.
As we are going to show explicitly for the dynamics at hand, 
we can control $N_{eff}$ by choosing different ROs, which opens considerable prospects for an efficient implementation of the simulations also when dealing with more complicated dynamics.

\subsection{Undriven master equation}\label{sec:und}
We set at first 
$b(t)=0$ and
study the case where we have a P-divisible evolution even if one of
the rates is temporally negative, i.e., CP-divisibility is broken; we
suppose in particular that $\gamma_3(t) < 0$ for any $t>0$. Then
$\Lambda_t$ is P-divisible provided that $\gamma_1(t),\gamma_2(t) \geq
|\gamma_3(t)|$. Dissipativity requires the stronger condition
\cite{KS} that $\gamma_1(t),\gamma_2(t) \geq
2|\gamma_3(t)|$. Interestingly, it turns out that whenever
$\gamma_3(t) < 0$ the map $\mathcal{J}_t(\rho) = \sum_{k=1}^3
\gamma_k(t) \sigma_k \rho \sigma_k/2$ is not positive. However, if the
generator (\ref{ENM}) gives rise to a P-divisible evolution,
one can define different positive $\mathbf{R}$-ROQJ by means
of different operators $\mathbf{C}(t)$ in Eqs.(\ref{eq:jprime})-(\ref{eq:hprime}).

\subsubsection*{Unraveling with $\mathbf{R1}$}
The first choice we make is to use
\begin{equation}
\mathbf{C}(t) = \frac{\gamma(t)}{2} \oper
\end{equation}
with
$$
\gamma(t) = \sum_k \gamma_k(t),
$$
as the corresponding map
\begin{equation}\label{eq.21} 
\mathbf{J}_t(\rho) = \mathcal{J}_t(\rho)+ \frac{\gamma(t)}{2} \rho
\end{equation} 
is positive, as shown in Appendix \ref{app:pr1r2}.
In particular
for the eternally non-Markovian evolution \cite{ENM,Nina}
defined by $\gamma_1=\gamma_2=1$, and $\gamma_3(t) = - {\rm tanh}\,t$,
one has a P-divisible evolution, but the corresponding generator
$\mathcal{L}_t^\dagger = \mathcal{L}_t$ is not dissipative.
Nevertheless, the map $\mathbf{J}_t$ is positive and one can define
the positive rate operator
$\mathbf{R1}_{\psi(t)} =\mathbf{J}_t(|\psi(t)\>\<\psi(t)|) \geq 0$.
The deterministic evolution is fixed by
$K1(t)=\frac{i}{2}\gamma(t)\mathbbm{1}$.  Figure \ref{fig1} (a) shows
7 example pure-state trajectories, with initial state
$|\psi(0)\rangle=\sqrt{0.1} | 1 \rangle + \sqrt{0.9} |2\rangle$, and
the probability $\rho_{11}$ of the state $|1\rangle$.  
Figure \ref{fig1} (b)
displays the final distribution of the Bloch vector X and Z components over the trajectories
while the inset shows the agreement between the
analytical results and simulations for the coherence $\rho_{12}$.

\begin{figure}[t]
\centering
\includegraphics[width=0.99\linewidth]{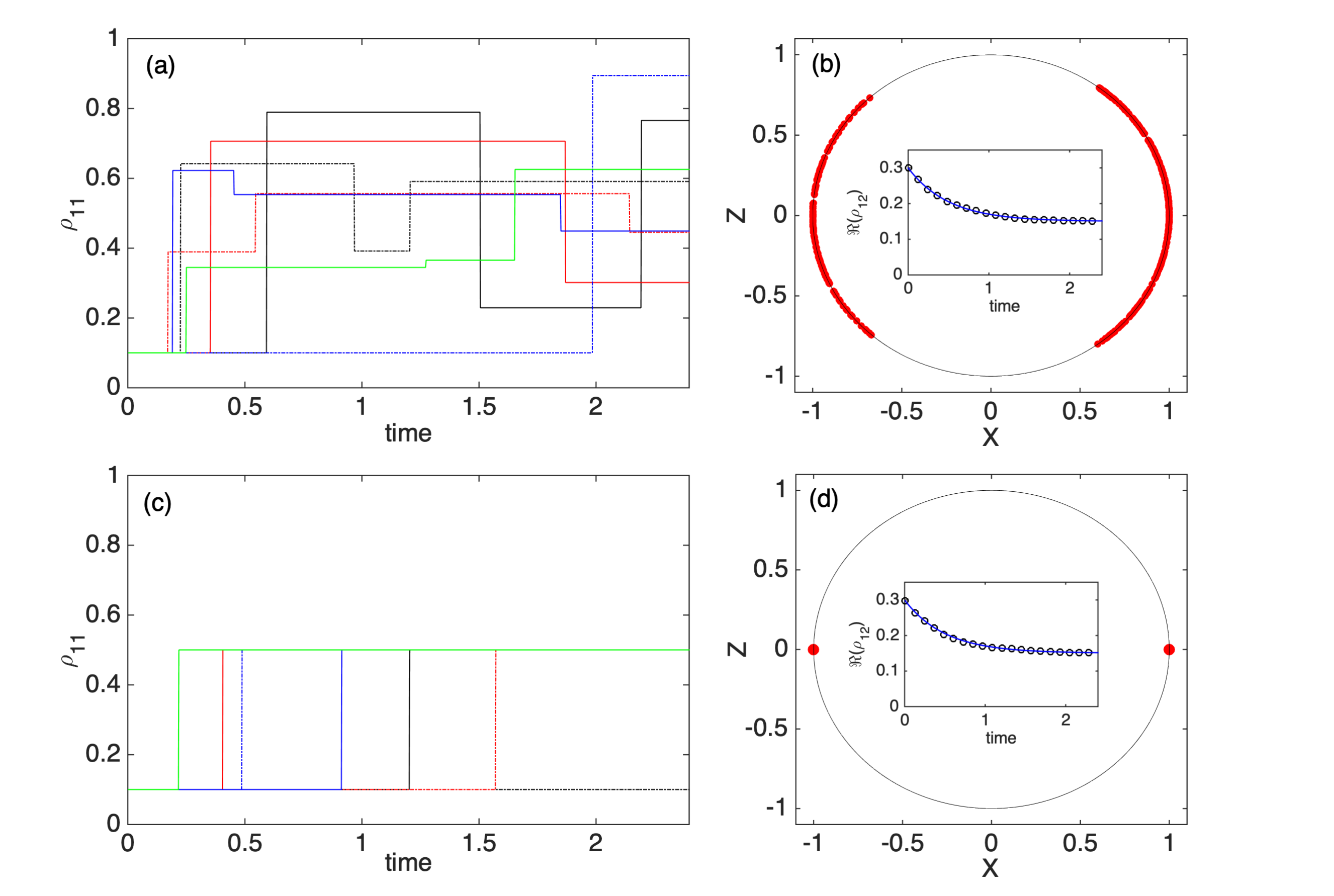}
\caption{Simulation results using rate operator $\mathbf{R1}$ [(a), (b)] and
$\mathbf{R2}$  [(c), (d)] for model~\eqref{ENM} without
driving. Example trajectories are displayed in (a) and (c).
Final steady state distributions for Bloch vector X and Z components are shown in  (b) and (d).
The insets in (b) and (d) demonstrate the agreement between the analytical results (solid line) and simulation results (dots) with $10^4$ trajectories. 
The error bars are similar to the  size of the dots.}
\label{fig1}
\end{figure}

\subsubsection*{Unraveling with $\mathbf{R2}$}
The second choice we make takes advantage of the
general result for qubits dynamics in Proposition~\ref{prop:jj}. 
In fact, as we show explicitly in Appendix~\ref{app:fnm}, the eternal non-Markovian evolution satisfies the
assumptions of the mentioned proposition, meaning
that it is possible to define a $\mathbf{R}$-ROQJ with a fixed
set of (deterministically evolving) states after the
jumps. Actually, in Appendix \ref{app:fnm} we define a continuous
family of rate operators with fixed post-jump
states for the eternal non-Markovian dynamics,
also ensuring their positivity.
Quite interestingly, this means that there exists a continuous family of
continuous measurement schemes for the same open-system dynamics
originating from the non-unique decomposition (\ref{eq:llp}) of the generator. 
We stress that this freedom is inherently different from the well-known \cite{Breuer} 
unitary freedom in defining the operators that appear
in the master equation~(\ref{GKLS}).
Moreover, the instruments defining the different measurement schemes,
as well as the associated probabilities, will generally be different, while the post-measurement
states will be the same.

A special instance of the $\mathbf{R}$-ROQJs with fixed post-jump
states is given by
\begin{equation}\label{eq:R2}
\mathbf{R2}_{\psi(t)} = \mathbf{R1}_{\psi(t)} -\gamma_3(t) \ket{\psi(t)}\bra{\psi(t)}
\end{equation}
associated with the deterministic evolution $K2(t)=\frac{i}{2}[\gamma_1(t) +\gamma_2(t)]\mathbbm{1}$.
The corresponding simulation results are shown in Fig.~\ref{fig1} (c) and (d).
The example trajectories in Fig.~\ref{fig1} (c) illustrate that the
ensemble consists now a discrete set of pure states.
In the final
distribution of states, see Fig.~\ref{fig1} (d), we have only two
states $|\pm\rangle=\frac{1}{\sqrt{2}}(|1\rangle \pm |2\rangle )$. As
a matter of fact, during the whole simulation, only three states
appear: the initial state and $|\pm\rangle$; in other terms, in this case the 3 states
fixing the unraveling according to Proposition~\ref{prop:jj} are
even time independent. 
Remarkably, $\mathbf{R2}_{\psi(t)}$
explicitly demonstrates that the measurement basis can be fixed once
and for all, without any state and time dependence.

\subsubsection*{Unraveling with $\mathbf{R3}$ and $\mathbf{W}$}
\begin{figure}[t]
\centering
\includegraphics[width=0.99\linewidth]{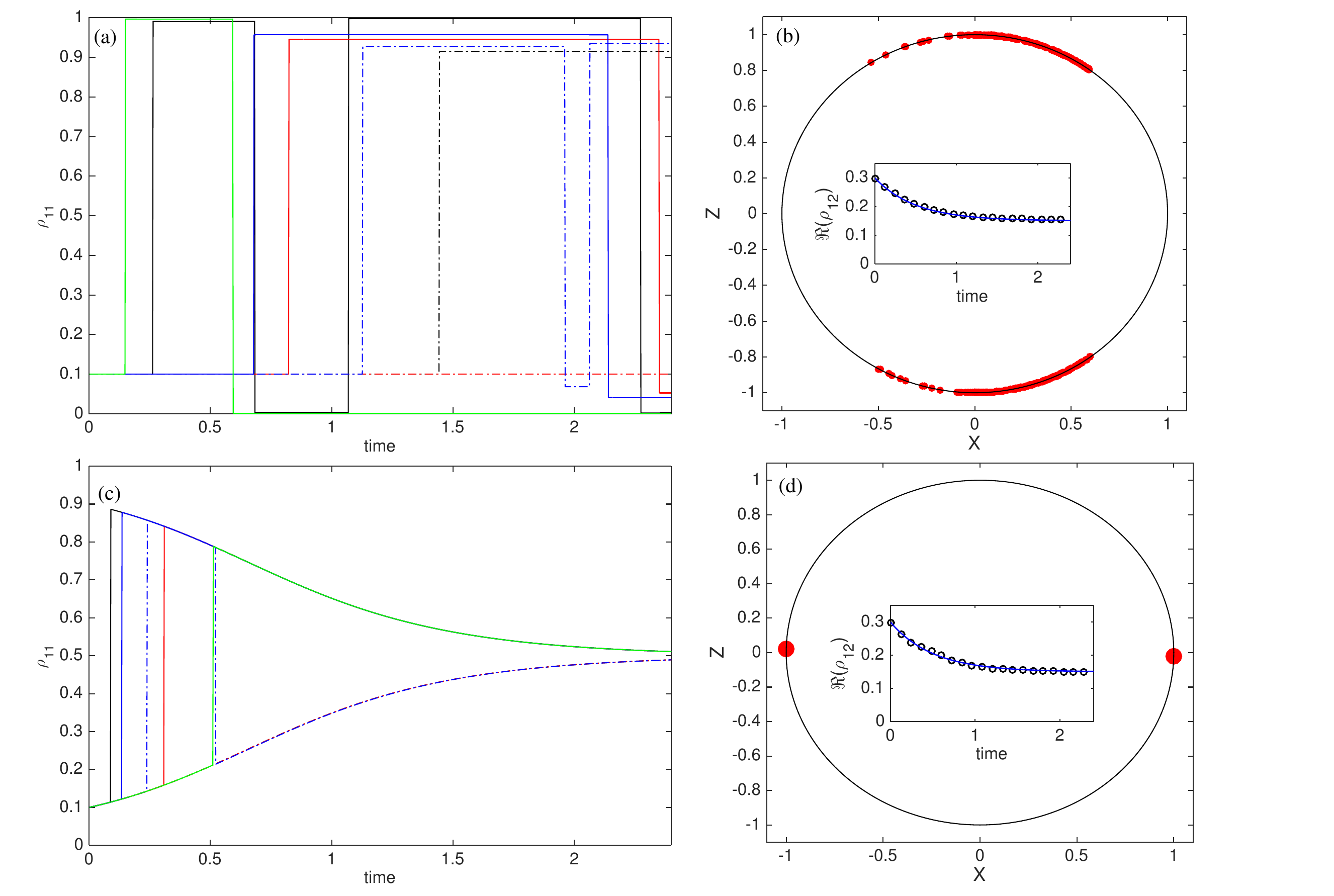}
\caption{Simulation results using rate operator $\mathbf{R3}$ [(a), (b)] and
$\mathbf{W}$  [(c), (d)]. Seven example pure state trajectories are displayed in (a) and (c)
by plotting the probability of state $|1\rangle.$
Final steady state distributions for Bloch vector X and Z components, with 500 trajectories, are shown in  (b) and (d).
Here, each dot corresponds to a single trajectory. In (d), all the 500 dots have the value $Z=0$ and $X=-1/2$ or $X=+1/2$,
i.e., we have only two distinctively visible dots. The insets in (b) and (d) demonstrate the agreement between the analytical and simulation results by using the coherence $\rho_{12}$ and having $10^4$ trajectories. In the insets, the solid line is the analytical result and the dots are the simulation results. The error bars are similar to the  size of the dots.
In all the panels, the initial state is  $|\psi(0)\rangle=\sqrt{0.1} | 1 \rangle + \sqrt{0.9} |2\rangle$ and the used time step size is $dt=0.002$.
  }
\label{figS1}
\end{figure}

Finally, we consider the rate operator
\begin{equation}
\label{eq:R3}
 \mathbf{R3}_\psi= \frac{1}{2}\sum_{k=1}^3 \gamma_k(t) \sigma_k  |\psi(t)\rangle\langle\psi(t)| \sigma_k -  \frac{1}{2} \gamma_3(t) |\psi(t)\rangle\langle\psi(t)|,
\end{equation}
which is positive as shown in Appendix \ref{app:pr1r2}.
The corresponding deterministic evolution is fixed by the linear operator $K3(t)=\frac{i}{4}[\gamma_1(t) +\gamma_2(t)]\mathbbm{1}$.
We show the simulation results in Fig.~\ref{figS1} (a) and (b).
The example realizations in  Fig.~\ref{figS1} (a) demonstrate that, similarly to $\mathbf{R1}$ (cf.~Fig.~\ref{fig1}),
the jumps continue even though the steady state has been reached. However, contrary to $\mathbf{R1}$,
now it holds all the time that $\rho_{11} \leqslant 0.1$ or  $\rho_{11} \geqslant 0.9$. This can be seen more clearly in
Fig.~\ref{figS1} (b). While the distribution of trajectories in the steady state for  $\mathbf{R1}$ contained arcs on the western and eastern sides of the
XZ-plane of the Bloch sphere, now with $\mathbf{R3}$ the distribution covers arcs on the northern and southern sides of the circle.
As before, there is excellent agreement between the analytical and simulation results [see the inset of Fig.~\ref{figS1} (b).]

The last operator we consider is  $\mathbf{W}$, see Eq.~(\ref{ro}).
We display the simulation results in Fig.~\ref{figS1} (c) and (d).
With the example realizations in Fig.~\ref{figS1} (c), one can clearly see that now also the deterministic evolution changes the states.
Moreover, the jumps happen between a pair of states only and terminate when the steady state is reached.
 Fig.~\ref{figS1} (d) shows that similarly to  $\mathbf{R2}$ [cf.~Fig.~\ref{fig1} (d)], all the trajectories eventually end up being on one of two states on the equator of the Bloch sphere. However, how they reach these points is totally different with respect to  $\mathbf{R2}$.
In terms of the measurement scheme, there is another crucial difference between  $\mathbf{W}$ and $\mathbf{R2}$.
With $\mathbf{W}$ the post-measurement states are time dependent, while with $\mathbf{R2}$ they are time independent.

Table \ref{tabS1} collects and compares the properties of all the used four operators in terms of
(i) whether the quantum jumps continue also in the asymptotic regime or terminate when the steady state is reached, (ii) whether the deterministic evolution changes the state of the trajectories, (iii) whether the post-measurement states are time dependent or time independent, and (iv) what is the effective ensemble size in the simulation (how many different kinds of state vectors the ensemble consists of point-wise in time).
Overall, the rate operator $\mathbf{R2}$ has the most appealing properties from simulation and fundamental interpretation points of views.

\begin{table*}[t]
\begin{center}
\begin{tabular}{ccccc}
   \hline
         Rate
       & Asymptotic
       & Deterministic
       & Time-independent
       &Effective ensemble \\
       operator
       & jumps?
       & changes?
       & post-measurement states?
       &size\\

       \hline
     $\mathbf{R1}$& yes  & no  & no & $\infty$\\
       $\mathbf{R2}$& no & no   & yes  & 3 \\
      $\mathbf{R3}$ &yes & no   & no  &  $\infty$\\
      $\mathbf{W}$& no & yes   & no   & 2

     \end{tabular}
     \caption[t2]{\label{tabS1}
The basic characterization of the four used rate operators for stochastic simulations.
Quantum jumps can either continue or terminate when steady state is reached. Depending on the form
of the effective Hamiltonian giving the deterministic evolution, the state of the trajectory can either change or remain unchanged after renormalization. The states after the measurements in the continuous-measurement scheme can be either time dependent or time independent.
The effective ensemble size describes how many different kinds of state vectors the ensemble consists of point-wise in time.}
\end{center}
\end{table*}

\subsection{Driven master equation}\label{sec:dri}
As a second example, we add a time-dependent driving $b(t)$ to
the evolution~(\ref{ENM})
and 
consider the decay rates $\gamma_1=\gamma_2=1$
and
$\gamma_3(t)=-\frac{1}{2}\tanh t$. For the 
driving, we choose
\begin{equation}
b(t) = C+\int\limits_0^t\mathrm{d}s\, \Omega(s;\mu,\sigma)
\end{equation}
including an integral over the Gaussian 
function 
$$
\Omega(s;\mu,\sigma)=\frac{1}{\sqrt{2\pi}\sigma}e^{-(s-\mu)^2/(2\sigma^2)}
$$ 
and a constant 
$$
C=\int_{-\infty}^0\mathrm{d}s\,\Omega(s;\mu,\sigma).
$$ 
At time $t=0$, $b(0)=C$ and $b(t)\to 1$ when 
$t\gg (\mu+\sigma)$;
the amplitude of the drive is ramped up from the initial value $C$ to the asymptotic value 1 
over a time-scale fixed by $\sigma$, essentially modeling a finite time quench of the Hamiltonian for the open system.
\begin{figure}[t]
\centering
\includegraphics[width=0.95\linewidth]{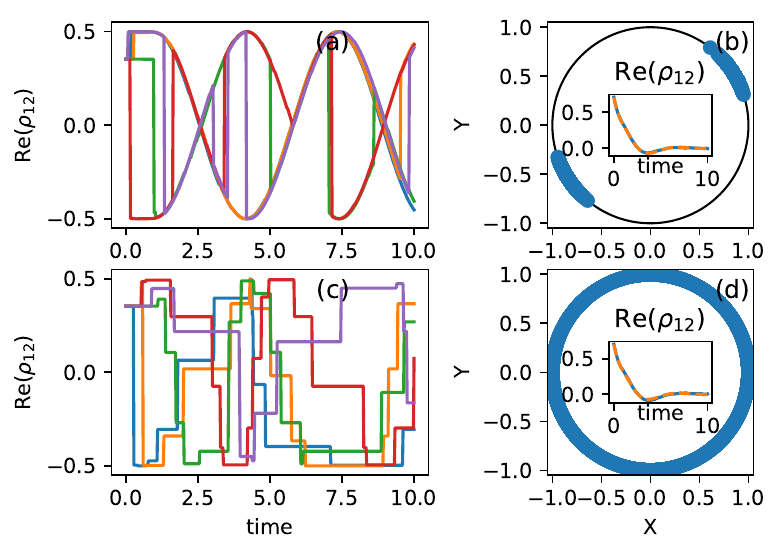}
\caption{Simulation results using rate operator
$\mathbf{R1}$ [(a), (b)] and $\mathbf{R1'}$ [(c), (d)] for the model
~(\ref{ENM}) with time dependent driving. 
Example trajectories are displayed in (a) and (c).
Final steady state distributions for Bloch
vector X and Y components are shown in (b)
and (d).  The
insets in (b) and (d) demonstrate the agreement between the analytical
results (solid line)
and simulation results (dots), with $10^4$ trajectories.
}
\label{fig2}
\end{figure}
The driving does not affect the positivity of the rate operators $\mathbf{R1}$ and 
$\mathbf{R2}$.  However, we use the decomposition~(\ref{eq:jprime}) and define a new 
rate operator
\begin{align}\label{eq:R1tprime}
  \mathbf{R1'}_{\psi(t)}=\mathbf{R1}_{\psi(t)}+i\frac{b(t)}{2}[\sigma_z,\vert{\psi(t)}\rangle\langle{\psi(t)}\vert],
\end{align}
which fully takes into account the driving, so that the deterministic evolution between the 
jumps will not depend on it. In other terms, as announced earlier, 
we can implement the simulation of time-dependent coherent driving 
with pure jump dynamics (provided that the 
rate operator $\mathbf{R1'}$ is positive).

In Fig.~\ref{fig2} we plot the dynamics for
$\mu=1$, $\sigma=1/4$ and for an initial state $\ket{\psi_0}=\cos\frac{\pi}{8}|1\rangle+\sin\frac{\pi}{8}|2\rangle$.
In panel (a) the coherent driving is clearly visible in the deterministic evolution between the 
jumps when using the rate operator $\mathbf{R1}$. In contrast, in 
panel (c) the sample trajectories do not have any deterministic evolution between the jumps 
since the driving is absorbed by the  
rate operator $\mathbf{R1'}$. In the insets of panels (b) and (d), we verify that we indeed unravel
the dynamics of Eq.~(\ref{ENM}).
The distribution of the states in the long time limit can be understood by looking at
$\varphi(t)=\arctan\frac{y(t)}{x(t)}$, i.e., the angle between the 
$x$- and $y$-axis of the Bloch sphere, measured from the positive $x$-axis. 
In the long time limit, when $b(t)\approx 1$ and all of the trajectories
have almost reached the equator of the Bloch sphere, each jump 
just shifts the phase $\varphi(t_+)=\varphi(t_-)\pm\frac{\pi}{4}$, where
times $t_\pm$ are just after and before the jump, respectively. 
The initial transient period randomizes the phases and the resulting smeared
distribution for the phase, as seen in panel (d) of Fig.~\ref{fig2}, occurs.

\begin{figure}
  \includegraphics[width=0.95\linewidth]{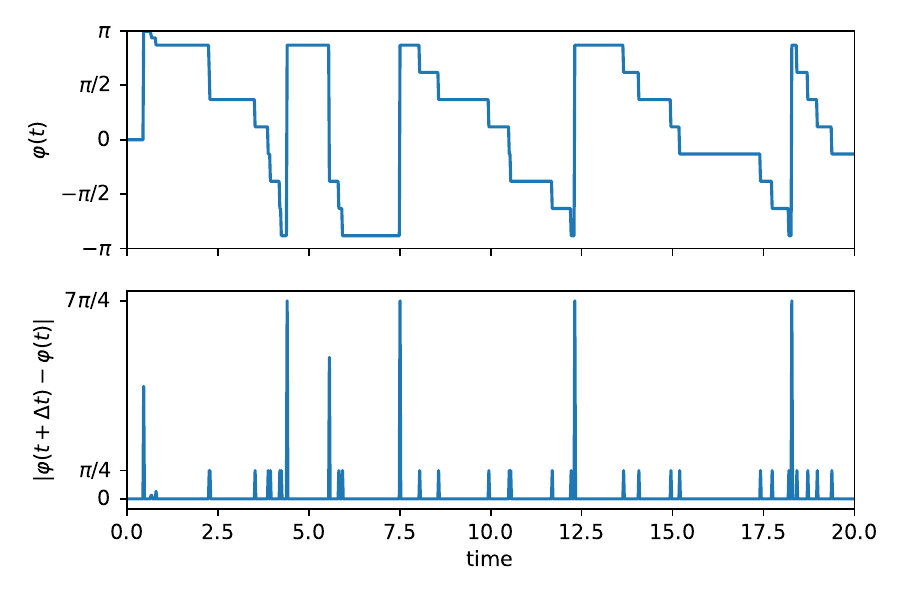}
  \caption{\label{fig:long_driven_trajectory}
   The phase evolution (upper panel) and phase jumps (lower panel) for
  a single trajectory with time dependent driving. We can see that initially 
  there is a transient where the phase jumps are not multiples of $\pi/4$. This 
  leads to the smearing of the ensemble in the long time limit observed in Fig.~\ref{fig2}}.
\end{figure}
\begin{figure}
  \includegraphics[width=0.95\linewidth]{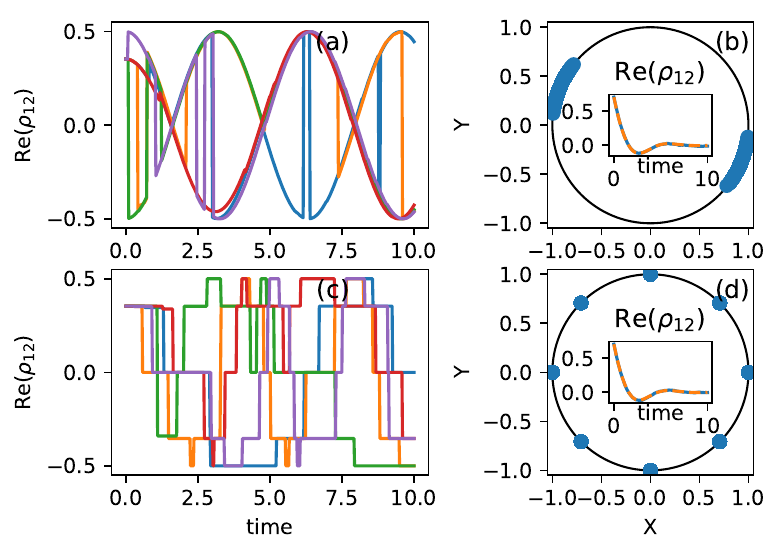}
  \caption{\label{fig:time_independent_driving} Time independent driving
    for $b(t)=1$. In panel a) we plot the dynamics using rate operator
  $\mathbf{R1}$ and the driving term is clearly visible in the real part 
  of the coherences. In panel b) we verify in the inset the validity with the 
  simulations and plot the distribution of Bloch vectors in the ensemble at 
  at the final time. In paneld c) we use the rate operator $\mathbf{R1'}$ and 
  the trajectories contain only jumps. In panel d) we plot the 
  Bloch vectors of the ensemble at the final time. Clearly, the phase
  is rational multiple of $\pi/4$ for every member of the ensemble.}
\end{figure}
For further confirmation, we plot in Figure~\ref{fig:long_driven_trajectory}
$\varphi(t)$ for a
single trajectory. 
We clearly see that after a transient period the phase changes during the jumps
in steps $\frac{\pi}{4}$, which can be understood as follows.
Let $z^2=\cos^2\theta$ be the squared $z$-component of a unit Bloch 
vector in spherical coordinates. The rate operators $\mathbf{R1}$ and $\mathbf{R1'}$ both induce a 
map $z^2\mapsto\frac{\alpha^2z^2}{(2-\alpha)^2}$ when $t\to\infty$. Hence, 
the jumps push the dynamics on the equatorial plane of the Bloch sphere.  Since 
the coherent driving preserves the $z$-component, the pure states of the 
ensemble will be distributed along the equator of the Bloch sphere in the long time 
limit. Any such pure state can be written as 
$\vert\psi\rangle = \frac{1}{\sqrt{2}}(|1\rangle+e^{i\varphi}|2\rangle)$.
The eigenstates of the rate operator $\mathbf{R1'}_{\psi}$ in the long time limit are
$|\phi_{\psi,\pm}\rangle=\frac{1}{2}(|1\rangle \pm e^{-i\varphi}\xi|2\rangle)$, where 
$\xi=\frac{1+ib}{\sqrt{1+b^2}}$. Importantly, $|\xi|=1$ and the
relative phase of the eigenstates is $e^{i\varphi}\xi$. 
In the long time limit, the effect of the quantum jumps corresponds 
to phase jumps  $e^{i\phi}\mapsto\pm e^{i\phi}\xi$.
In the long time limit $b(t)=1$.
Interestingly, the phase jumps for the rate operator $\mathbf{R1'}$ 
in the special case $b=1$ are rational multiples of $\pi$ in 
the long time limit. Namely, 
$\xi=e^{i\pi/4}$ which is the 8th root of unity. 

In Fig.~\ref{fig:time_independent_driving} we plot the dynamics for 
time-independent driving $b(t)=1$ and especially in panel d) we can see 
how in this case the definite phase relations are preserved. The phase $\varphi(t)$ aquires
values which are rational multiples  of $\pi/4$  for every trajectory of the ensemble.

\section{Conclusions}\label{sec:con} 
We introduced a class of
quantum-jump unravelings (denoted as $\mathbf{R}$-ROQJ) that
 interpolates between MCWF for CP-divisible dynamics and the rate-operator
 unraveling for P-divisible evolutions \cite{Caiaffa2017,Jyrki-2020}.
$\mathbf{R}$-ROQJ guarantees a well-defined
continuous-measurement scheme for the class of quantum evolutions defined by 
the dissipativity condition  (\ref{KS-L}), which
is strictly stronger than P-divisibility.
Furthermore, $\mathbf{R}$-ROQJ 
takes full advantage of the freedom one has in dividing the master equation into
a determinisitic and a jump part, by setting a desired time-independent linear non-Hermitian Hamiltonian,
even in the case of an external driving, or selecting a fixed post-jump basis.

The approach put forward here will likely be useful also to deal with more general classes
of non-Markovian dynamics,
by including reversed jumps in the unraveling \cite{Piilo2008,Piilo2009}. 
In \cite{Jyrki-2020}, reversed jumps for the standard ROQJ were defined, 
thus showing that they can be integrated effectively with ROQJ. 
Even more, the combination of reversed jumps and $\mathbf{R}$-ROQJ 
will allow us to fully exploit the freedom in manipulating the deterministic and jump parts of the unraveling,
without having to restrict to rate operators with positive eigenvalues.

Taking into account more complex, higher dimensional open-system dynamics, it will
be possible to perform a detailed comparison of the efficiency of our method
with respect not only to different jump-based unravelings, but also to other techniques
to solve general master equations.
In addition, the definition of different unravelings based on distinct master-equation representations 
associated with the very same dynamics
might lead to novel insights to connect
the jumpless part of the evolution of an open quantum system to 
effective non-Hermitian Hamiltonians and their engineering \cite{Naghiloo2019,Minganti2019,Minganti2020,Ashida2020,Chen2021,Roccati2022}.
Finally, in future work  we will consider the diffusive limit of
  $\mathbf{R}$-ROQJ unravelings; with a similar approach as in
  Ref.~\cite{Luoma2020}, we could interpolate between
  piecewise deterministic and continuous quantum trajectories for
  the $\mathbf{R}$-ROQJ, and also between the ROQJ quantum trajectory descriptions
  presented in ~\cite{Caiaffa2017,Jyrki-2020}

\vspace{1cm}

\begin{acknowledgments}
We would like to thank Lajos Di{\'o}si for very interesting discussions
on the topic and for drawing our attention on the papers \cite{Braun2000,Braun2002}.
D.C. was supported by the Polish National Science
Centre project 2018/30/A/ST2/00837.
J.P. acknowledges support from Magnus Ehrnrooth Foundation.
A.S. acknowledges funding by the FFABR project of MIUR and PSR-2 2020 by UNIMI.
\end{acknowledgments}

\bibliographystyle{unsrtnat}

\begin{thebibliography}{11}

\bibitem{Breuer} H.-P. Breuer and F. Petruccione,
\href{https://doi.org/10.1093/acprof:oso/9780199213900.001.0001}{{\em The Theory of Open Quantum Systems}} (Oxford Univ. Press,
Oxford, 2007).
\bibitem{Carmichael} H.J. Carmichael,\emph{An Open System Approach to Quantum Optics},
 \href{https://doi.org/10.1007/978-3-540-47620-7}{Lectures Notes in Physics (Springer, Berlin, 1993)}.
\bibitem{Dalibard1992} J. Dalibard, Y. Castin, and K. M{\o}lmer, 
\href{https://doi.org/10.1103/PhysRevLett.68.580}{Phys. Rev. Lett. {\bf68}, 580 (1992)}.
\bibitem{Basche1995}T. Basche, S. Kummer, and C. Brauchle, 
\href{https://doi.org/10.1038/373132a0}{Nature {\bf373}, 132 (1995)}.
\bibitem{Peil1999}S. Peil and G. Gabrielse, 
\href{https://doi.org/10.1103/PhysRevLett.83.1287}{Phys. Rev. Lett. {\bf83}, 1287 (1999)}.
\bibitem{Jelezko2002}F. Jelezko, I. Popa, A. Gruber, C. Tietz, J. Wrachtrup, A. Nizovtsev, and S. Kilin, \href{https://doi.org/10.1063/1.1507838}{Appl. Phys. Lett. {\bf81}, 2160 (2002)}.
\bibitem{Gleyzes2007} S. Gleyzes, S. Kuhr, C. Guerlin, J. Bernu, S. Del{\'e}glise, U.B. Hoff, M. Brune, J.-M. Raimond, and S. Haroche, 
\href{https://doi.org/10.1038/nature05589}{Nature {\bf446}, 297 (2007)}.
\bibitem{Vijay2011}R. Vijay, D. H. Slichter, and I. Siddiqi, 
\href{https://doi.org/10.1103/PhysRevLett.106.110502}{Phys. Rev. Lett. {\bf106}, 110502 (2011)}.
\bibitem{Minev2019}Z. K. Minev, S. O. Mundhada, S. Shankar, P. Reinhold, R. Guti{\'e}rrez-J{\'a}uregui, R.J. Schoelkopf, M. Mirrahimi, H. J. Carmichael, and M.H.~Devoret, 
\href{https://doi.org/10.1038/s41586-019-1287-z}{Nature {\bf570}, 200 (2019)}. 
\bibitem{Plenio1998} M. B. Plenio and P. L. Knight, 
\href{https://doi.org/10.1103/RevModPhys.70.101}{Rev. Mod. Phys. {\bf70}, 101 (1998)}.
\bibitem{Daley2014} A.J. Daley, 
\href{https://doi.org/10.1080/00018732.2014.933502}{Adv. Phys. {\bf 63}, 77 (2014)}.
\bibitem{Percival2002} I.Percival, \emph{Quantum State Diffusion} (Cambridge University Press, Cambridge, England, 2002)
\bibitem{Barchielli2009} A. Barchielli and M. Gregoratti, \emph{Quantum Trajectories and Measurements in Continuous Time: The Diffusive Case}, 
\href{https://doi.org/10.1007/978-3-642-01298-3}{Lecture Notes in Physics 782 (Springer, Berlin, 2009)}.
\bibitem{Wiseman1993} H.M. Wiseman and G.J. Milburn, 
\href{https://doi.org/10.1103/PhysRevA.47.1652}{Phys. Rev. A {\bf47}, 1652 (1993)}. 
\bibitem{Strunz1999} W. T. Strunz, L. Di\'osi, and N. Gisin, 
\href{https://doi.org/10.1103/PhysRevLett.82.1801}{Phys. Rev. Lett. {\bf82}, 1801 (1999)}.
\bibitem{Yu1999} T. Yu, L. Di\'osi,  N. Gisin, and W. T. Strunz,  
\href{https://doi.org/10.1103/PhysRevA.60.91}{Phys. Rev. A {\bf60}, 91 (1999)}.
\bibitem{Luoma2020} K. Luoma, W.T. Strunz, and J. Piilo, 
\href{https://doi.org/10.1103/PhysRevLett.125.150403}{Phys. Rev. Lett. {\bf 125}, 150403 (2020)}.
\bibitem{Siddiqi2013} K. W. Murch, S. J. Weber, C. Macklin, and  I.~Siddiqi, 
\href{https://doi.org/10.1038/nature12539}{Nature {\bf 502}, 211 (2013)}. 
\bibitem{Huard2016} P. Campagne-Ibarcq, P. Six, L. Bretheau, A.~Sarlette, M. Mirrahimi, P. Rouchon, and B. Huard, 
\href{https://doi.org/10.1103/PhysRevX.6.011002}{Phys. Rev. X {\bf 6}, 011002 (2016)}.  
\bibitem{Gourgy2016} S. Hacohen-Gourgy, L.S.~Martin, E. Flurin, V.V.~Ramasesh, K.B. Whaley, and 
I. Siddiqi, 
\href{https://doi.org/10.1038/nature19762}{Nature {\bf 538}, 491 (2016)}.
\bibitem{Ficheux2018} Q. Ficheux, S. Jezouin, Z. Leghtas, and  B. Huard, 
\href{https://doi.org/10.1038/s41467-018-04372-9}{Nat. Comm. {\bf 9}, 1926 (2018)}.
\bibitem{Barchielli1991} A. Barchielli and V.P. Belavkin, 
\href{https://doi.org/10.1088/0305-4470/24/7/022}{J. Phys. A: Math. Gen. {\bf24}, 1495 (1991)}.
\bibitem{Laine2010} E.-M. Laine, J. Piilo, and H.-P. Breuer, 
\href{https://doi.org/10.1103/PhysRevA.81.062115}{Phys. Rev. A {\bf81}, 062115 (2010)}.
\bibitem{Chruscinski2011} D. Chrusci{\'n}ski, A. Kossakowski, and {\'A}. Rivas, 
\href{https://doi.org/10.1103/PhysRevA.83.052128}{Phys. Rev. A {\bf 83}, 052128 (2011)}.
\bibitem{Rivas2012} {\'A}. Rivas and S. F. Huelga, 
\href{https://doi.org/10.1007/978-3-642-23354-8}{\emph{Open Quantum Systems} (Springer, New York, 2012)}.
\bibitem{Rivas2010} {\'A}. Rivas, S. F. Huelga, and M. B. Plenio, 
\href{https://doi.org/10.1103/PhysRevLett.105.050403}{Phys. Rev. Lett. {\bf105}, 050403 (2010)}.
\bibitem{Rivas2014}{\'A}. Rivas, S. F. Huelga, and M. B. Plenio, 
\href{https://doi.org/10.1088/0034-4885/77/9/094001}{Rep. Prog. Phys. {\bf77}, 094001 (2014)}.
\bibitem{Breuer2009} H.-P. Breuer, E.-M. Laine, and J. Piilo, 
\href{https://doi.org/10.1103/PhysRevLett.103.210401}{Phys. Rev. Lett. {\bf103}, 210401 (2009)}.
\bibitem{Breuer2016} H.-P. Breuer, E.-M. Laine, J. Piilo, and B.~Vacchini, 
\href{https://doi.org/10.1103/RevModPhys.88.021002}{Rev. Mod. Phys. {\bf88}, 021002 (2016)}.
\bibitem{Piilo2008} J. Piilo, S. Maniscalco, K. H{\"a}rk{\"o}nen, and K.A.~Suominen, 
\href{https://doi.org/10.1103/PhysRevLett.100.180402}{Phys. Rev. Lett. {\bf100}, 180402 (2008)}.
\bibitem{Piilo2009} J. Piilo, K. H{\"a}rk{\"o}nen, S. Maniscalco, and K.A.~Suominen, 
\href{https://doi.org/10.1103/PhysRevA.79.062112}{Phys. Rev. A {\bf79}, 062112 (2009)}.
\bibitem{Gambetta2003} J. Gambetta and H.M. Wiseman, 
\href{https://doi.org/10.1103/PhysRevA.68.062104}{Phys. Rev. A {\bf68}, 062104 (2003)}.
\bibitem{Diosi2008} L. Di{\'o}si, 
\href{https://doi.org/10.1103/PhysRevLett.100.080401}{Phys. Rev. Lett. {\bf100}, 080401 (2008)}.
\bibitem{Wiseman2008} H.M. Wiseman and J.M. Gambetta, 
\href{https://doi.org/10.1103/PhysRevLett.101.140401}{Phys. Rev. Lett. {\bf101}, 140401 (2008)}.
\bibitem{Jyrki-2020} A. Smirne, M. Caiaffa, and J. Piilo, 
\href{https://doi.org/10.1103/PhysRevLett.124.190402}{Phys. Rev. Lett. {\bf 124}, 190402 (2020)}.
\bibitem{Diosi1985} L. Di{\'o}si, 
\href{https://doi.org/10.1016/0375-9601(85)90342-1}{Phys. Lett. A {\bf112}, 288 (1985)}.
\bibitem{Diosi1986} L. Di{\'o}si, 
\href{https://doi.org/10.1016/0375-9601(86)90692-4}{Phys. Lett. A {\bf114}, 451 (1986)}.
\bibitem{Diosi1988} L. Di{\'o}si, 
\href{https://doi.org/10.1088/0305-4470/21/13/013}{J. Phys. A {\bf21}, 2885 (1988)}.
\bibitem{Gisin1990} N. Gisin, 
\href{http://doi.org/10.5169/seals-116244}{Helv. Phys. Acta {\bf63}, 929 (1990)}.
\bibitem{Vacchini2011} B. Vacchini, A. Smirne, E.-M. Laine, J. Piilo, H.P.~Breuer, 
\href{https://doi.org/10.1088/1367-2630/13/9/093004}{New J. Phys. {\bf13}, 093004 (2011)}.
\bibitem{Sabrina-PRL} D. Chru\'sci\'nski and S. Maniscalco, 
\href{https://doi.org/10.1103/PhysRevLett.112.120404}{Phys. Rev. Lett. {\bf 112}, 120404 (2014)}.
\bibitem{Wissmann2015} S. Wi{\ss}mann, H.-P. Breuer, B. Vacchini, 
\href{https://doi.org/10.1103/PhysRevA.92.042108}{Phys. Rev. A {\bf92}, 042108 (2015)}.
\bibitem{Wiseman2010} H. M. Wiseman and G. J. Milburn, 
\href{https://doi.org/10.1017/CBO9780511813948}{\emph{Quantum Measurement and Control}} (CUP, Cambridge, 2010).
\bibitem{Zhang2017} J. Zhangab, Y.-X. Liu, R.-B. Wuab, K. Jacobs, and F. Nori,
\href{https://doi.org/10.1016/j.physrep.2017.02.003}{Phys. Rep. {\bf 679}, 1 (2017)}.
\bibitem{Gourgy2018} S. Hacohen-Gourgy, L. P. Garc{\`i}a-Pintos, L.S.~Martin, J. Dressel, 
and I. Siddiqi, 
\href{https://doi.org/10.1103/PhysRevLett.120.020505}{Phys. Rev. Lett. {\bf 120}, 020505 (2018)}.
\bibitem{Martin2020} L.S. Martin, W.P. Livingston, S. Hacohen-Gourgy, H.M. Wiseman and
I. Siddiqi, 
\href{https://doi.org/10.1038/s41567-020-0939-0}{Nat. Phys. {\bf 16}, 1046 (2020)}.
\bibitem{Magrini2021} L. Magrini, P. Rosenzweig, C. Bach, A.~Deutschmann-Olek,
S.G. Hofer, S. Hong, N. Kiesel, A. Kugi, and
M. Aspelmeyer, 
\href{https://doi.org/10.1038/s41586-021-03602-3}{Nature {\bf 595}, 373 (2021)}.
\bibitem{L}  G. Lindblad, 
\href{https://doi.org/10.1007/BF01608499}{Comm. Math. Phys. {\bf 48}, 119 (1976)}.
\bibitem{GKS} V. Gorini, A. Kossakowski, and E.C.G.~Sudarshan, 
\href{https://doi.org/10.1063/1.522979}{J. Math. Phys. {\bf 17}, 821 (1976)}.
\bibitem{Chruscinski2010} D. Chrusci{\'n}ski, and A. Kossakowski, 
\href{https://doi.org/10.1103/PhysRevLett.104.070406}{Phys. Rev. Lett. {\bf 104}, 070406 (2010)}.
\bibitem{Caiaffa2017} M. Caiaffa, A. Smirne, and A. Bassi, 
\href{https://doi.org/10.1103/PhysRevA.95.062101}{Phys. Rev. A {\bf95}, 062101 (2017)}.
\bibitem{Braun2000} T.A. Brun, 
\href{https://doi.org/10.1103/PhysRevA.61.042107}{Phys. Rev. A {\bf 61}, 042107 (2000)}.
\bibitem{Braun2002} T.A. Brun, 
\href{https://doi.org/10.1119/1.1475328}{Am. J. Phys. {\bf 70}, 719 (2002)}.
\bibitem{Diosi2017} L. Di{\'o}si, 
\href{https://doi.org/10.1088/1751-8121/aa6263}{J.Phys. A {\bf 50}, 16LT01 (2017)}.
\bibitem{ENM} M.J.W. Hall, J.D. Cresser, L. Li, and E.~Andersson, 
\href{https://doi.org/10.1103/PhysRevA.89.042120}{Phys. Rev. A {\bf 89}, 042120 (2014)}.
\bibitem{Chruscinski2015} D. Chru{\'s}ci{\'n}ski and F.A. Wudarski, 
\href{https://doi.org/10.1103/PhysRevA.91.012104}{Phys. Rev. A {\bf91}, 012104 (2015)}.
\bibitem{Nina} N. Megier, D. Chruscinski, J. Piilo, and W.T.~Strunz, 
\href{https://doi.org/10.1038/s41598-017-06059-5}{Sci. Rep. {\bf 7},  6379 (2017)}.
\bibitem{Heinosaari2012} T. Heinosaari and M. Ziman, 
\href{https://doi.org/10.1017/CBO9781139031103}{\emph{The Mathematical Language of Quantum Theory}}, (Cambridge University Press, Cambridge, 2012).
\bibitem{Wiseman1996} H.M. Wiseman, 
\href{https://doi.org/10.1088/1355-5111/8/1/015}{Quantum Semiclass. Opt. {\bf 8}, 205 (1996)}.
\bibitem{Paulsen} V. Paulsen, 
\href{https://doi.org/10.1017/CBO9780511546631}{{\em Completely Bounded Maps and Operator
Algebras}} (Cambridge University Press, Cambridge, 2003).
\bibitem{Stormer} E. St{\o}rmer, 
\href{https://doi.org/10.1007/978-3-642-34369-8}{\emph{Positive Linear Maps of Operator Algebras}},
Springer Monographs in Mathematics (Springer, New York,
2013).
\bibitem{Molmer1996} K. M{\o}lmer and Y. Castin, 
\href{https://doi.org/10.1088/1355-5111/8/1/007}{Quantum Semiclass. Opt. {\bf 8}, 49 (1996)}.
\bibitem{KS} D. Chru\'sci\'nski and F. Mukhamedov, 
\href{https://doi.org/10.1103/PhysRevA.100.052120}{Phys. Rev. A. {\bf 100}, 052120 (2019)}.
\bibitem{Naghiloo2019} M. Naghiloo, M. Abbasi, Yogesh N. Joglekar, and K. W. Murch, 
\href{https://doi.org/10.1038/s41567-019-0652-z}{Nat. Phys. {\bf 15}, 1232 (2019)}.
\bibitem{Minganti2019} F. Minganti, A. Miranowicz, R. W. Chhajlany, and F. Nori, 
\href{https://doi.org/10.1103/PhysRevA.100.062131}{Phys. Rev. A {\bf 100}, 062131 (2019)}.
\bibitem{Minganti2020} F. Minganti, A. Miranowicz, R. W. Chhajlany, I. I. Arkhipov, and F. Nori,
\href{https://doi.org/10.1103/PhysRevA.101.062112}{Phys. Rev. A {\bf 101}, 062112 (2020)}.
\bibitem{Ashida2020} Y. Ashida, Z. Gong, and M. Ueda,
\href{https://doi.org/10.1080/00018732.2021.1876991}{Adv. Phys. {\bf 69}, 3 (2020)}.
\bibitem{Chen2021} W. Chen, M. Abbasi, Y. N. Joglekar, and K. W. Murch,
\href{https://doi.org/10.1103/PhysRevLett.127.140504}{Phys. Rev. Lett. {\bf 127}, 140504 (2021)}.
\bibitem{Roccati2022} F. Roccati, G.M. Palma, F. Bagarello, and F. Ciccarello
\href{https://doi.org/10.1142/S1230161222500044}{Op. Syst. Inf. Dyn. {\bf 29}, 2250004 (2022)}.
\end{thebibliography}

\appendix
\numberwithin{equation}{section}
\onecolumngrid

\section{Proof of Proposition \ref{prop:jj}}
Consider $\mathcal{H}_S =\mathbbm{C}^2$ and a master equation as in Eq.(\ref{GKLS});
moreover assume that there is a basis $\left\{\ket{\varphi_1},\ket{\varphi_2}\right\}$
and a linear operator $\mathbf{C}(t)$ on $\mathbbm{C}^2$
such that Eqs.(\ref{eq:cond2}) and (\ref{eq:cond1}) hold.
First note that Eq.(\ref{eq:cond2}), along with the Hermiticity preservation condition, imply
(we neglect from now on the time dependence)
$J^{k l}_{i j} = J^{l k}_{j i}$. Using the latter, it is 
readily shown that the action of $\mathcal{J}'_t$ in Eq.(\ref{eq:jprime})
with $\mathbf{C}$ defined as in Eq.(\ref{eq:a}) can be written as
\begin{equation}
  \mathcal{J}'_t(\overline{\rho})= \mathcal{J}_t(\overline{\rho}) 
  + \frac{1}{2}\left(\mathbf{C} \overline{\rho} +\overline{\rho} \mathbf{C}^\dagger\right) 
  =  \left( \begin{array}{cc} 
\alpha& \beta \\
\beta  & \alpha
\end{array} \right) \quad  
\begin{cases}
      \alpha = & (2 J^{12}_{11}+y)\rho_{12}+J^{22}_{11}\rho_{22}+J^{11}_{22}\rho_{11}\\
      \beta  = & (J^{11}_{12}+\frac{y}{2}+J^{12}_{11}+J^{12}_{22})\rho_{11}
      +(J^{22}_{12}+\frac{y}{2})\rho_{22} \\
      & + \frac{1}{2}(2 J^{12}_{12}+2 J^{21}_{12}+J^{11}_{22}-J^{11}_{11}+J^{22}_{11}-J^{22}_{22})\rho_{12}
    \end{cases} \label{eq:conaa}
\end{equation}
for any state $\overline{\rho}$ such that 
\begin{equation}\label{eq:cond12}
\rho_{12}=\rho_{21},
\end{equation}
where indeed $\rho_{ij} = \bra{\varphi_i}\overline{\rho}\ket{\varphi_j}$
(i.e., there are no phases with respect to the selected basis).

Now, the crucial observation is that the matrix in Eq.(\ref{eq:conaa}) has \emph{fixed} eigenvectors:
for any values of $\alpha$ and $\beta$, they are in fact given by
\begin{equation}\label{eq:eigenv}
\ket{\varphi_{\pm}} = \frac{1}{\sqrt{2}}\left(\ket{\varphi_1}\pm \ket{\varphi_2}\right),
\end{equation}
which directly implies the statement to be proven, whenever
the rate operator is positive if referred to pure states satisfying Eq.(\ref{eq:cond12}):
\begin{itemize}
\item as said, we choose an initial state $\ket{\psi_0}$ that satisfies Eq.(\ref{eq:cond3});
\item but then, due to Eq.(\ref{eq:cond1}), the state $\ket{\psi_0(t_1)}$
before the first jump at time $t_1$ (recall that $D(t,s)$ is defined in Eq.(\ref{eq:det}))
$$
\frac{D(t_1,0) \ket{\psi_0}}{\left\|D(t_1,0) \ket{\psi_0}\right\|}
$$
will be such that Eq.(\ref{eq:cond12}) holds,
so that the jump operator 
$\mathbf{R}_{\psi_0(t_1)} = \mathcal{J}'_t(\ket{\psi_0(t_1)}\bra{\psi_0(t_1)})$ fixed by Eq.(\ref{eq:jprime})
with $\mathbf{C}(t)$ as in Eq.(\ref{eq:a}) has the two eigenvectors $\ket{\varphi_{\pm}}$
in Eq.(\ref{eq:eigenv});
\item thus, after the jump (that is well defined due to the positivity of the rate operator) the system will be either in $\ket{\varphi_+}$ or in $\ket{\varphi_-}$;
\item since also the latter are in the form as in Eq.(\ref{eq:cond3})
the state before the second jump,
$$
\frac{D(t_2,t_1) \ket{\varphi_\pm}}{\left\|D(t_2,t_1) \ket{\varphi_\pm}\right\|},
$$
will still satisfy Eq.(\ref{eq:cond12}),
so that the state after the second jump will be
either $\ket{\varphi_+}$ or $\ket{\varphi_-}$, and so on.
\end{itemize}
All in all, we have the 3 possible families of states
$\left\{\ket{\psi_0(t)}, \ket{\varphi_{+}(t,s)},\ket{\varphi_-(t,s)}\right\}$
that are obtained from $\left\{\ket{\psi_0}, \ket{\varphi_{+}},\ket{\varphi_-)}\right\}$ 
via the deterministic evolution in Eq.(\ref{eq:det}); the instants of occurrence of the jumps
will determine which parts of the deterministic evolution will be actually involved in each trajectory.

\section{Proof of Proposition \ref{prop:ks}}\label{app:propks}

Defining
\begin{equation}
  \mathbf{K}(t) = \int_{U(N)} \mathcal{L}_t^\dagger(U^\dagger)U dU ,
\end{equation}
one finds for an arbitrary $V \in U(N)$
\begin{equation}
  \int_{U(N)} \mathcal{L}_t^\dagger(VU^\dagger)U dU =  \int_{U(N)} \mathcal{L}_t^\dagger({U'}^\dagger) U' V dU' =   \mathbf{K}(t)V ,
\end{equation}
with $U' = UV^\dagger$. Hence for an arbitrary system operator $X$ one has
\begin{equation}
  \int_{U(N)} \mathcal{L}_t^\dagger(XU^\dagger)U dU =  \mathbf{K}(t)X .
\end{equation}

Now we use the dissipativity condition, Eq.(\ref{KS-L}),
\begin{equation}
    \mathcal{L}^\dagger_t(Y^\dagger Y) \geq \mathcal{L}^\dagger_t(Y^\dagger)Y + Y^\dagger \mathcal{L}^\dagger_t(Y) ,
\end{equation}
for $Y = UX$, with $U \in U(N)$. It leads to
\begin{eqnarray}
    \mathcal{L}^\dagger_t(Y^\dagger Y) \geq  \mathcal{L}^\dagger_t(X^\dagger U^\dagger) UX + X^\dagger U^\dagger \mathcal{L}^\dagger_t(UX)  
    = \mathcal{L}^\dagger_t(X^\dagger U^\dagger) UX + X^\dagger [\mathcal{L}^\dagger_t(X^\dagger U^\dagger) U]^\dagger .
\end{eqnarray}
Averaging over $U(N)$ yields 
\begin{eqnarray}
  \int_{U(N)}  \mathcal{L}^\dagger_t(Y^\dagger Y) dU \geq  \int_{U(N)} \mathcal{L}^\dagger_t(X^\dagger U^\dagger) UdU\, X + + X^\dagger \left( \int_{U(N)}\mathcal{L}^\dagger_t(X^\dagger U^\dagger) U DU \right)^\dagger = \mathbf{K}(t)X^\dagger X +  X^\dagger X \mathbf{K}^\dagger(t) ,
\end{eqnarray}
and hence defining a linear map as
\begin{equation}
  \mathbf{J}^\dagger_t(X) =  \mathcal{L}_t^\dagger(X) - (\mathbf{K}(t) X + X \mathbf{K}^\dagger(t)) ,
\end{equation}
we have shown that
\begin{equation}
  \mathbf{J}^\dagger_t(X^\dagger X) \geq 0,
\end{equation}
which proves that the map $\mathbf{J}_t$ is positive. 

\section{Positivity of the rate operators $\mathbf{R1}_{\psi(t)} $,  $\mathbf{R2}_{\psi(t)}$ and $\mathbf{R3}_{\psi(t)}$}
\label{app:pr1r2}

Here, we prove the positivity of $\mathbf{R1}_{\psi(t)} $, $\mathbf{R2}_{\psi(t)}$, and  $\mathbf{R3}_{\psi(t)}$ defined in Sec.\ref{sec:ex}, fixed by Eqs.(\ref{eq.21}), (\ref{eq:R2}) and (\ref{eq:R3}), assuming only that the $\gamma_k(t)$ satisfy P-divisibility condition, that is, $\gamma_i(t) + \gamma_j(t) \geq 0$ for $i \neq j = 1,2,3$ and that $\gamma_3(t) < 0$, so that P-divisibility means that
$$   \gamma_k(t) + \gamma_3(t) \geq 0 \ , \ k=1,2. $$
Note that the eternal non-Markovian dynamics considered in the main text is a special case of this
class of dynamics.

First, we note that the qubit generator
\begin{equation}
  \mathcal{L}_t(\rho) = \frac 12 \sum_{k=1}^3 \gamma_k(t) (\sigma_k \rho \sigma_k - \rho),
\end{equation}
may be rewritten as
\begin{equation}
  \mathcal{L}_t(\rho) =  \sum_{k=1}^3 \gamma_k(t) (J_k(\rho) - \rho),
\end{equation}
where
\begin{equation}
  J_k(\rho) = \frac 12 (\sigma_k \rho \sigma_k + \rho) ,
\end{equation}
are CPTP for $k=1,2,3$. 
We then have the following result.

\begin{Proposition}
The  map (cf. Eq.~(\ref{eq.21})) 
\begin{equation}
  \mathbf{J}_t(\rho) = \sum_{k=1}^3 \gamma_k(t) J_k(\rho)  = \mathcal{J}_t(\rho) + \frac{\gamma(t)}{2} \rho
\end{equation}
is positive. 
\end{Proposition}

 The proof is based on the following observation: for a qubit system a linear map $\Phi : M_2(\mathbb{C}) \to M_2(\mathbb{C})$ is positive if and only the corresponding Choi matrix $C_\Phi$
$$   C_\Phi := \sum_{i,j=1}^2 |i\>\<j | \otimes \Phi(|i\>\<j|) $$
satisfies
$$   C_\Phi  = A + (\mathbbm{1}\otimes T) B , $$
where $A$ and $B$ are positive $4 \times 4$ matrices, and `$\mathbbm{1} \otimes T$' denotes partial transposition. The Choi matrix for the map $\mathbf{J}_t$ reads
\begin{equation}
C_{\mathbf{J}_t} =  \left(
\begin{array}{cc|cc}
 {\gamma_1}+{\gamma_2}+2 {\gamma_3} & 0 & 0 & {\gamma_1}+{\gamma_2} \\
 0 & {\gamma_1}+{\gamma_2} & {\gamma_1}-{\gamma_2} & 0 \\   \hline
 0 & {\gamma_1}-{\gamma_2} & {\gamma_1}+{\gamma_2} & 0 \\
 {\gamma_1}+{\gamma_2} & 0 & 0 & {\gamma_1}+{\gamma_2}+2 {\gamma_3} \\
\end{array}
\right) ,
\end{equation}
where we skipped the time dependence. Note, that $ C_{\mathbf{J}_t}  = A + (\mathbbm{1}\otimes T) B$, where

\begin{equation}
A =  \left(
\begin{array}{cc|cc}
 {\gamma_1}+{\gamma_2}+2 {\gamma_3} & 0 & 0 & {\gamma_1}+{\gamma_2}+ 2 \gamma_3 \\
 0 & {\gamma_1}+{\gamma_2}+ 2 \gamma_3 & {\gamma_1}-{\gamma_2} & 0 \\   \hline
 0 & {\gamma_1}-{\gamma_2} & {\gamma_1}+{\gamma_2} +  2 \gamma_3 & 0 \\
 {\gamma_1}+{\gamma_2} + 2 \gamma_3 & 0 & 0 & {\gamma_1}+{\gamma_2}+2 {\gamma_3} \\
\end{array}
\right) ,
\end{equation}
and

\begin{equation}\label{Bg}
B = -2 \gamma_3 \left(
\begin{array}{cc|cc}
 0 & 0 & 0 & 0\\
 0 & 1 & 1 & 0 \\   \hline
 0 & 1 & 1 & 0 \\
 0 & 0 & 0 & 0 \\
\end{array}
\right) .
\end{equation}
Since $\gamma_3(t) < 0$ one has $B\geq 0$. Now, $A \geq 0$ if and only if

\begin{equation}\label{||}
  {\gamma_1}+{\gamma_2}+ 2 \gamma_3 \geq |{\gamma_1}-{\gamma_2}| .
\end{equation}
Let us consider two scenario: 

\begin{enumerate}
  \item if $\gamma_1 \geq \gamma_2$, then (\ref{||}) is equivalent to $\gamma_2 + \gamma_3 \geq 0$,
  \item conversely, if $\gamma_2 \geq \gamma_1$, then (\ref{||}) is equivalent to $\gamma_1 + \gamma_3 \geq 0$,
\end{enumerate}
which ends the proof of positivity of $A$. Now, since $A$ and $B$ are positive, the map $\mathbf{J}_t$ is positive and hence $\mathbf{R1}_{\psi(t)} \geq 0$.\\

Positivity of $\mathbf{R2}_{\psi(t)}$ is evident due to the relation
\begin{equation}\label{}
  \mathbf{R2}_{\psi(t)} = \mathbf{R1}_{\psi(t)} - \gamma_3(t) |\psi(t)\>\<\psi(t)| ,
\end{equation}
and $\gamma_3(t) < 0$. Finally, the operator  $\mathbf{R3}_{\psi(t)}$ corresponds to the map $\frac 12 ( \sum_k \gamma_k \sigma_k \rho \sigma_k - \gamma_3 \rho)$. Its Choi matrix reads

\begin{equation}
C =  \left(
\begin{array}{cc|cc}
 0 & 0 & 0 & -2 {\gamma_3} \\
 0 & {\gamma_1}+{\gamma_2} & {\gamma_1}-{\gamma_2} & 0 \\   \hline
 0 & {\gamma_1}-{\gamma_2} & {\gamma_1}+{\gamma_2} & 0 \\
 -2 {\gamma_3} & 0 & 0 & 0 \\
\end{array}
\right)  = A' +  (\mathbbm{1}\otimes T) B' ,
\end{equation}
where

\begin{equation}
A' =  \left(
\begin{array}{cc|cc}
 0 & 0 & 0 & 0 \\
 0 & {\gamma_1}+{\gamma_2}+ 2 \gamma_3 & {\gamma_1}-{\gamma_2} & 0 \\   \hline
 0 & {\gamma_1}-{\gamma_2} & {\gamma_1}+{\gamma_2} +  2 \gamma_3 & 0 \\
 0 & 0 & 0 & 0 \\
\end{array}
\right) ,
\end{equation}
and $B'=B$ defined by (\ref{Bg}). The proof of positivity of $A'$ is very similar to that of $A$.

\section{Fixed post-jump states for the eternal non-Markovian dynamics}\label{app:fnm}
Here, we first show explicitly that the dynamics defined by the master
equation~(\ref{ENM}) is actually a special case of the two-level system dynamics
identified by Prop.\ref{prop:jj}, so that we can define a $\mathbf{R}$-ROQJ unraveling
with fixed post-jump states via Eq.(\ref{eq:a}). We will then discuss the positivity of such an unraveling.

Consider the basis $\left\{\ket{1},\ket{0}\right\}$
of eigenvectors of $\sigma_z$; in such a basis the coefficients of
$\mathcal{J}_t$ according to Eq.(\ref{eq:166}) read
\begin{equation}
\left( \begin{array}{c|cc} 
J^{11}_{1 1}(t) &  \gamma_3(t)/2 \\
J^{11}_{1 0}(t) &0 \\
J^{11}_{0 1}(t) &  0 \\
J^{11}_{0 0}(t) & (\gamma_1(t)+\gamma_2(t))/2 \\
J^{10}_{1 1}(t) &  0 \\
J^{10}_{1 0}(t) & -\gamma_3(t)/2 \\
J^{10}_{0 1}(t) &  (\gamma_1(t)-\gamma_2(t))/2 \\
J^{10}_{0 0}(t) & 0 \\
J^{01}_{1 1}(t) &  0 \\
J^{01}_{1 0}(t) &  (\gamma_1(t)-\gamma_2(t))/2\\
J^{01}_{0 1}(t) & -\gamma_3(t)/2  \\
J^{01}_{0 0}(t) & 0 \\
J^{00}_{1 1}(t) &  (\gamma_1(t)+\gamma_2(t))/2 \\
J^{00}_{1 0}(t) &   0\\
J^{00}_{0 1}(t) &0  \\
J^{00}_{0 0}(t) & \gamma_3(t)/2
\end{array} \right) 
\end{equation}
and hence (setting $x=0$) the operator $\mathbf{C}(t)$ as in Eq.(\ref{eq:a}) 
is 
\begin{eqnarray}\label{aah}
\mathbf{C}(t) &=& \left( \begin{array}{cc} 
\frac{\gamma_1(t)+\gamma_2(t)-\gamma_3(t)}{2} & y(t) \\
y(t)&  \frac{\gamma_1(t)+\gamma_2(t)-  \gamma_3(t)}{2}
\end{array} \right) = \frac{\gamma_1(t)+\gamma_2(t)-  \gamma_3(t)}{2}\mathbbm{1} +y(t) \sigma_x;
\end{eqnarray}
while the operator fixing the deterministic evolution according to Eq.(\ref{eq:det}) is
\begin{equation}\label{eq:detextra}
D(t,s) = \exp\left(-\frac{1}{2}\int_s^t d \tau (\gamma_1(\tau)+\gamma_2(\tau)
)\mathbbm{1} + y(\tau)\sigma_x \right);
\end{equation}
note that for $y(t)=0$ we recover the rate operator $\mathbf{R2}_{\psi(t)}$ defined by Eq.(\ref{eq:R2}).
It is easy to see that the assumptions of Prop.\ref{prop:jj} in Eq.(\ref{eq:cond2}) and (\ref{eq:cond1})
hold, so that, starting from an initial state as in Eq.(\ref{eq:cond3}), the $\mathbf{R}$-ROQJ
fixed by $\mathbf{C}(t)$ as in Eq.(\ref{aah}) consists of jumps among 3 deterministically evolving states,
whenever the positivity of the rate operator is ensured. In particular, 
for a continuous family of $\mathbf{R}$-ROQJ this is the case whenever
the dynamics is P-divisible, as shown below.
 
First, recall that the P-divisibility
condition for the eternal non-Markovian dynamics is $\gamma_1(t),\gamma_2(t)\geq |\gamma_3(t)|$ for all $t\geq0$. Moreover, going back to the proof of Prop.~\ref{prop:jj}
in Sec.\ref{app:propks}, the positivity of the rate operator is fixed by the eigenvalues
$$
\alpha(t)\pm \beta(t)
$$
of $\mathcal{J}'_t(\ket{\psi(t)}\bra{\psi(t)})$, see the first equality in Eq.(\ref{eq:conaa}).
Actually, rather than studying the eigenvalues, it is convenient to look at the trace and determinant of the corresponding
$2\times 2$ matrix: for any pure state of the form $\ket{\psi}= c\ket{\varphi_1}\pm\sqrt{1-c^2}\ket{\varphi_2}$ (i.e., as in Eq.(\ref{eq:cond3}))
we then have the two conditions
\begin{eqnarray}
\alpha(t) &=& \frac{\gamma_1(t) + \gamma_2(t)}{2} \pm c \sqrt{1-c^2} y(t) \geq 0 \nonumber\\
\alpha^2(t)-\beta^2(t) &=& \left(\frac{\gamma_1(t) + \gamma_2(t)}{2} \pm c \sqrt{1-c^2} y(t)\right)^2-
\left(\frac{y(t)}{2} \pm c \sqrt{1-c^2}(\gamma_1(t)-\gamma_3(t))\right)^2 \geq 0 \label{eq:mort}.
\end{eqnarray}
Both inequalities hold for any $-1\leq c \leq 1$ if 
$0\leq y(t)\leq \gamma_1(t) + \frac{1}{2}\gamma_2(t)- \frac{1}{2}\gamma_3(t)$
for all $t\geq 0$. The validity of the first inequality directly follows from the fact that
$\gamma_1(t),\gamma_2(t) \geq 0$ and $\gamma_1(t) + \frac{1}{2}\gamma_2(t)- \frac{1}{2}\gamma_3(t)\leq \gamma_1(t) + \gamma_2(t)$  due to P-divisibility and that the maximum
value of $|c| \sqrt{1-c^2}$ for $-1\leq c \leq 1$ is $1/2$.
For the second inequality, let us consider first the case where the term $\pm c$ appearing in it is positive,
so that the inequality becomes
$$
\gamma_1(t)+\gamma_2(t)-y(t) \geq 2|c|\sqrt{1-c^2} \left(\gamma_1(t)-\gamma_3(t)-y(t)\right)
$$
as  $\gamma_1(t)-\gamma_3(t)\geq 0$ due to P-divisibility
and we used that $\alpha(t)\geq 0$, as well as $y(t)\geq 0$;
but the maximum value with respect to $c$ of the term at the r.h.s. is $\max\left\{0, \gamma_1(t)-\gamma_3(t)-y(t) \right\}$, so that the inequality holds for any
$y(t)\leq  \gamma_1(t) + \gamma_2(t)$ ($\gamma_2(t) \geq -\gamma_3(t)$ due to P-divisibility).
Instead, if $\pm c$ is negative the second inequality in (\ref{eq:mort}) becomes
$$
\frac{\gamma_1(t)+\gamma_2(t)}{2} - |c|\sqrt{1-c^2}y(t) \geq \left|\frac{y(t)}{2}-|c|\sqrt{1-c^2}(\gamma_1(t)-\gamma_3(t)) \right|.
$$
If $y(t)\geq2|c|\sqrt{1-c^2}(\gamma_1(t)-\gamma_3(t))$, we have
$
\gamma_1(t)+\gamma_2(t) - y(t) \geq 2 |c|\sqrt{1-c^2}(y(t) - (\gamma_1(t)-\gamma_3(t))),
$
whose r.h.s. maximum value with respect to $c$ is $\max\left\{0, y(t) - (\gamma_1(t)-\gamma_3(t)) \right\}$
so that the inequality holds for any $-1\leq c \leq 1$ if  
$y(t) \leq \gamma_1(t) + \frac{1}{2}\gamma_2(t)- \frac{1}{2}\gamma_3(t)$.
If $y(t)<2|c|\sqrt{1-c^2}(\gamma_1(t)-\gamma_3(t))$, we have
$
\gamma_1(t)+\gamma_2(t) + y(t) \geq 2 |c|\sqrt{1-c^2}(y(t) + (\gamma_1(t)-\gamma_3(t))),
$
whose r.h.s. maximum is $\max\left\{0, y(t) + (\gamma_1(t)-\gamma_3(t)) \right\}$ 
so that the inequality is ensured by $\gamma_2(t)\geq - \gamma_3(t)$, 
as well as $\gamma_1(t), \gamma_2(t), y(t) \geq 0$.


In total, for any $y(t)$ such that 
$0\leq y(t)\leq \gamma_1(t) + \frac{1}{2}\gamma_2(t)- \frac{1}{2}\gamma_3(t)$ for all $t\geq 0$
the rate operator $\mathbf{R}_{\psi} = \mathcal{J}'_t(\ket{\psi}\bra{\psi})$ is positive
for any $\ket{\psi}=c \ket{\varphi_1}\pm\sqrt{1-c^2}\ket{\varphi_2}$
with $-1 \leq c \leq 1$
and then if the initial state satisfies the latter condition (which is preserved by the deterministic evolution
in Eq.(\ref{eq:detextra})), we will have a \emph{positive} rate-operator
unravelling with the states $\left\{\ket{\psi_0(t)}, \ket{\varphi_{+}(t,s)},\ket{\varphi_-(t,s)}\right\}$.

\end{document}